\documentclass[11pt]{article}
\usepackage[letterpaper,margin=1in]{geometry}
\usepackage[T1]{fontenc}
\usepackage{lmodern}
\usepackage{amsmath,amssymb,amsthm,mathtools}
\usepackage{microtype,graphicx}
\usepackage{xcolor}
\usepackage{booktabs,tabularx,longtable}
\usepackage{enumitem}
\usepackage{algorithm}
\usepackage{placeins}
\usepackage[noend]{algpseudocode}
\usepackage{tikz}
\usetikzlibrary{arrows.meta,positioning,calc,fit,decorations.pathreplacing}
\definecolor{figblue}{RGB}{37,92,130}
\definecolor{figorange}{RGB}{178,91,36}
\definecolor{figgreen}{RGB}{45,111,86}
\tikzset{paper arrow/.style={-{Stealth[length=2mm]},semithick},
  paper box/.style={draw,rounded corners=2pt,align=center,inner sep=5pt}}
\usepackage[colorlinks=true,linkcolor=blue!55!black,citecolor=blue!55!black,urlcolor=blue!55!black]{hyperref}
\numberwithin{equation}{section}
\newtheorem{theorem}{Theorem}[section]
\newtheorem{lemma}[theorem]{Lemma}
\newtheorem{proposition}[theorem]{Proposition}
\newtheorem{corollary}[theorem]{Corollary}
\theoremstyle{definition}

\theoremstyle{remark}

\newcommand{\ALG}{\operatorname{ALG}}
\newcommand{\OPT}{\operatorname{OPT}}
\newcommand{\ADV}{\operatorname{ADV}}
\newcommand{\cost}{\operatorname{cost}}

\newcommand{\E}{\mathbb E}

\newcommand{\eps}{\varepsilon}

\let\originalthebibliography\thebibliography
\renewcommand{\thebibliography}[1]{%
  \originalthebibliography{#1}%
  \setlength{\itemsep}{2pt}}

\title{Capacitated Line Aggregation with Maximum Delay}
\author{Tianhang Lu$^{\ast}$\qquad Runtian Ren$^{\ast}$\qquad Shengcai Liu\\[0.65em]
{\small Guangdong Provincial Key Laboratory of Brain-Inspired Intelligent Computation}\\
{\small Department of Computer Science and Engineering}\\
{\small Southern University of Science and Technology, Shenzhen 518055, China}\\
{\small \texttt{liusc3@sustech.edu.cn}}\\[0.3em]
{\footnotesize $^{\ast}$Equal contribution.}}
\date{}
\hypersetup{pdftitle={Capacitated Line Aggregation with Maximum Delay},pdfauthor={Tianhang Lu, Runtian Ren, Shengcai Liu}}
\begin{document}
\pagenumbering{roman}
\maketitle
\thispagestyle{empty}
\begin{abstract}
We study capacitated line aggregation with per-batch maximum delay.
Requests arrive over time at locations on a rooted half-line and are
served in instantaneous batches of at most $g$ requests. Each batch
pays the distance to its farthest request plus the waiting time of its
oldest request. The objective is to minimize the total cost.

We give a deterministic online algorithm with competitive ratio
$36+12\sqrt3$. The algorithm uses a Next Fit idea along paths of nested
time intervals: it assigns a new group to an existing batch when a
path-load test permits, and opens a new batch otherwise. The path load
controls both capacity and the cost of opening new batches.
Randomizing the common scale improves the expected competitive ratio
to $(18+6\sqrt3)/\ln2$ against an oblivious adversary. Both guarantees
extend to arbitrary indivisible request sizes.

We also study capacitated LAD in the learning-augmented model.
Advice revealed at its service times allows $(1+\eps)$-consistency against every complete
feasible advice schedule and $O(\eps^{-2})$-robustness. The analysis
charges repeated switches to successive decreases in the remaining
advice cost.
Controlled synthetic experiments with 9,180 exact-rational monitor runs
illustrate both the benefit of optimal advice and the delayed fallback
caused by the stated conservative constants. We provide all inputs,
outcomes, and reproducibility code; the experiments do not claim
performance on learned predictions or large-scale demand traces.
\end{abstract}

\clearpage
\pagenumbering{arabic}
\section{Introduction}
\label{sec:introduction}

Sending an acknowledgment, replenishing several types of goods, and
serving requests along a common route all allow several requests to
share one service cost. Waiting for more requests can reduce that cost,
but also delays the requests already present. This tradeoff has been
studied in TCP acknowledgment
\cite{dooly2001tcp,seiden2000guessing,karlin2001dynamic},
joint replenishment (JRP)
\cite{askoy1988multi,goyal1989joint,joneja1990joint,buchbinder2008online,bienkowski2014better},
and control-message aggregation
\cite{khanna2002control,brito2004competitive,bienkowski2013chain}.
Multi-level aggregation (MLA) extends these problems to rooted trees, with
algorithms for additive waiting costs and deadlines
\cite{bienkowski2020mlap,bienkowski2021new,buchbinder2017depth,
azar2019framework,mcmahan2021dcompetitive,turoczy2025memory}
and a stochastic-arrival variant~\cite{mari2024online}.

Under \emph{per-batch maximum delay}, each service pays its service cost
plus the waiting time of its oldest request, and the objective sums these
costs over all services. Lu et al.~\cite{lu2026maximum} study this objective
for uncapacitated multi-level aggregation on rooted trees. They obtain
optimal competitive ratios of $2$ deterministically and $e/(e-1)$ against
an oblivious adversary. We study what happens when a capacity constraint
is added, focusing on line aggregation with delay (LAD).

In capacitated LAD, requests arrive at locations on a rooted half-line.
Each service completes at most $g$ pending requests and pays the
distance to the farthest one, together with their maximum waiting time.
Besides choosing when to serve, the algorithm must choose which requests
to serve together. Even on a line, these decisions are closely coupled:
a distant request can pay for serving nearby requests, but each nearby
request occupies a place that a later request might use.

Under additive delay, every pending request contributes to the waiting
cost. Under maximum delay, younger requests can join a batch without
increasing its waiting cost. This sharing is useful only while the
requests fit. If an uncapacitated service is split into several feasible
batches, several of them may contain old requests, and each batch pays
its own maximum waiting time. Conversely, serving each distance class
independently avoids this splitting decision but loses opportunities
to share waiting across classes. For LAD, this gives a
logarithmic competitive ratio; see
Section~\ref{subsec:class-baseline} for the algorithm and a matching
lower-bound instance.

For LAD, we avoid this logarithmic loss by combining requests across
distance classes and obtain a constant-competitive online algorithm.
The algorithm assigns
requests to future batches through nested time intervals. Its admission
rule controls both the capacity of a batch and the cost of opening a
new batch when admission fails. Thus capacity does not force the
logarithmic loss incurred by independent distance classes.

\paragraph{Our contributions.}
We obtain the following main results in this paper.
\begin{itemize}[leftmargin=*,itemsep=0.4em,parsep=0pt,topsep=0.4em]
\item For LAD, we give a deterministic online algorithm with
competitive ratio $36+12\sqrt3<56.785$ for capacity $g$.
For radii on a fixed scaled dyadic grid, the competitive ratio is $18+6\sqrt3$.
A random common scale offset gives an expected competitive ratio of
$(18+6\sqrt3)/\ln2<40.962$ against an oblivious adversary. Both guarantees
extend to arbitrary indivisible request sizes, as described in
Section~\ref{subsec:arbitrary-sizes}.

\item We also study capacitated LAD in the learning-augmented model.
An algorithm receives advice only when a proposed service is due. We
combine this advice with any
deterministic $R$-competitive baseline. For $0<\eps\le1$, the resulting algorithm is
$(1+\eps)$-consistent against every complete feasible advice stream and
$(32/\eps^2+72/\eps+36)R$-robust. Our LAD baseline gives
$O(\eps^{-2})$-robustness.

\item We evaluate the stated advice monitor on 60 synthetic 12-request
instances, three capacities, and two feasible advice-corruption mechanisms,
giving 9,180 exact-rational runs. Exact offline optima and an independent
schedule evaluator support the cost comparisons. Optimal advice attains
the optimum on these tests, while conservative switching can make poor
advice substantially more expensive than the no-advice baseline. We retain
all outcomes and provide the full data and code in the ancillary material;
Section~\ref{sec:experiments} and Appendix~\ref{app:advice-experiments}
describe the protocol and its limitations.
\end{itemize}

At a single location, LAD reduces to capacitated TCP acknowledgment,
whose optimal deterministic competitive ratio is two. This special
case also gives a randomized lower bound of $e/(e-1)$ on competitive
ratios uniform in capacity. Appendix~\ref{sec:online} gives the
fixed-location algorithm and both lower-bound arguments.

\paragraph{Our techniques.}
Our algorithms address two difficulties: assigning requests to future
batches and switching between schedules that have served different requests.
\begin{itemize}[leftmargin=*,itemsep=0.5em,parsep=0.3em,topsep=0.4em]
\item \emph{A constant-competitive algorithm.}
We use Next Fit along paths of nested time intervals. A path-load
test accounts for group sizes and interval lengths, ensuring capacity
and paying for new batches when admission fails. Next Fit prevents
descendants from repeatedly charging the same ancestor.
Section~\ref{subsec:line-main-idea} develops the rule and its additive
charges through a worked example.

\item \emph{Online advice.}
Switching schedules requires clearing requests that the new schedule
may already have served. We separate the service and waiting costs
of this clearance. Baseline budgets control service costs; successive
decreases in the remaining advice cost pay waiting charges across
switches without repeatedly charging an entire advice batch.

\end{itemize}

\subsection{Related work}
\label{sec:related-work}

\paragraph{Aggregation with additive waiting and maximum delay.}
TCP acknowledgment, JRP, control-message aggregation, and MLA have been
studied with additive waiting costs and deadlines
\cite{dooly2001tcp,karlin2001dynamic,buchbinder2008online,
bienkowski2014better,bienkowski2013chain,bienkowski2020mlap,
bienkowski2021new,buchbinder2017depth,azar2019framework,
mcmahan2021dcompetitive,turoczy2025memory}.
Our direct predecessor is the work of Lu et al.~\cite{lu2026maximum}
on uncapacitated aggregation with per-batch maximum delay. They show
that an offline optimum serves consecutive arrival blocks. This yields
a polynomial-time offline dynamic program, which also sets the deadlines
of their online algorithm. With capacity, such a block may not fit in
one service, and splitting it can charge the same waiting period to
several batches.

Lu et al.~\cite{lu2026covering} extend the maximum-delay framework to
monotone subadditive service costs. With exact batch optimization, they
obtain optimal competitive ratios of $2$ deterministically and $e/(e-1)$
randomly against an oblivious adversary. Their framework
requires every request set to be realizable by one service. In LAD,
$f(Q)=\max_{r\in Q}x_r$ remains subadditive, but the union of two
capacity-feasible batches need not fit in one service. Treating their
partition as a single service would omit the separate maximum-delay
charges, so this framework does not directly apply.

\paragraph{Capacitated batching and capacity transfer.}
Lu et al.~\cite{lu2026binmaximum} study bin packing with per-bin maximum
delay and arbitrary positive item sizes. This is the single-location
special case of our size extension, with unit service cost. They obtain
an expected competitive ratio of approximately $2.418$ against an
oblivious adversary in the ideal random-real model. Our setting allows
requests at different distances to share a service.

Offline capacitated LAD with deadlines is captured by Nonner's
capacitated max-batching model~\cite{nonner2016batching}.
A request with window $[a_r,b_r]$ and radius $x_r$ corresponds to an
interval of weight $x_r$. A feasible batch contains at most $g$ intervals
with a common intersection and costs their maximum weight. The problem
is strongly NP-hard already for $g=3$ and admits a PTAS for every fixed
$g$. For proper interval graphs, Correa and Megow~\cite{correa2015clique}
give a PTAS with arbitrary capacity. These results concern deadline
feasibility and service cost, rather than per-batch maximum delay.

A different connection concerns adding capacity to an uncapacitated
schedule. Tour-partitioning bounds combine an uncapacitated service
cost with a demand-weighted root-distance lower bound
\cite{blauth2020routing,zhao2025capacity}. Fukunaga et
al.~\cite{fukunaga2014deliver} use a reduction from capacitated to
uncapacitated periodic inventory routing, where costs include inventory
holding. For aggregation, splitting a service at its original time
preserves additive waiting costs and deadline feasibility. Applying
the partitioning bounds and a volume lower bound therefore increases
the competitive ratio by only an additive constant; we give the
derivation in Appendix~\ref{sec:review}. This argument does not extend
to maximum delay, since each resulting batch pays its own maximum
waiting time.

\paragraph{Dynamic bin packing and busy-time scheduling.}
MinUsageTime dynamic bin packing and busy-time scheduling group items
or jobs under capacity constraints and charge the total time for which
bins or machines are used
\cite{li2014dynamic,kamali2015renting,li2016dynamic,tang2016firstfit,
ren2017competitiveness,flammini2010busy,shalom2014busy,mertzios2015busy,
ren2020heterogeneous}.
There is a temporal resemblance to our objective: the maximum waiting
time of a batch is the length of the union of its requests' waiting
intervals. In LAD, however, these intervals all end at a common service
time chosen by the algorithm. Capacity bounds the total number of
requests in the batch. In dynamic bin packing, departures free capacity
that later arrivals can reuse. Thus the shared busy-time interpretation
does not by itself transfer algorithms or lower bounds to LAD, which
also charges a distance-dependent service cost.

\paragraph{Learning-augmented LAD.}
Lu et al.~\cite{lu2026learninglad} study uncapacitated LAD with additive
waiting costs, including deterministic and randomized learning-augmented
algorithms using online suggestions of service lengths. We follow
their consistency benchmark of any feasible advised schedule. With
capacity, our advice also specifies which requests a service completes.
It is revealed only at the proposed service time. Our analysis must
account for requests that remain pending when the algorithm switches
schedules, even if the schedule it starts following has already served
them. Section~\ref{sec:online-advice} gives the switching rule and the
residual-advice argument that controls these costs.

\paragraph{Organization.}
Section~\ref{sec:model} introduces the model.
Section~\ref{sec:line-path-load} develops the constant-competitive LAD
algorithm, from the failure of independent distance classes to the
Next Fit construction, its analysis, and the arbitrary-size extension.
Section~\ref{sec:online-advice} gives the online-advice algorithm and
the route to its consistency and robustness bounds. The appendices contain
the fixed-location results, the baseline upper bound, capacity transfer,
the detailed path-load and advice proofs, and the arbitrary-size extension.

\section{Preliminaries}
\label{sec:model}
In LAD, each request $r$ has an arrival time $a_r\ge0$ and a distance
$x_r\ge0$ from the origin on a rooted half-line, both revealed at
arrival. We call $x_r$ its radius. Requests remain distinct even when
they arrive at the same location and time. The capacity is a known
integer $g\ge1$. A service $(Q,t)$ completes a nonempty set $Q$ of at
most $g$ pending requests at time $t$. Its service and total costs are
\begin{equation}\label{eq:maximum-objective}
 f(Q)=\max_{r\in Q}x_r,\qquad
 c(Q,t)=f(Q)+\max_{r\in Q}(t-a_r).
\end{equation}
The objective is to minimize the sum of these costs. Each request must
be served exactly once. Services are instantaneous, and several
separately paid batches may be served at the same time. Capacity is
shared across all locations in a batch; a completed batch cannot accept
later requests. At one location, LAD is capacitated TCP acknowledgment.

Inputs are finite, but their number of requests and end time are not
announced. All arrivals at a timestamp are revealed before services at
that time. An online algorithm uses only the revealed history and must
eventually serve every request. We write $\ALG(I)$ for its total cost on
input $I$ and $\OPT(I)$ for the minimum cost of a feasible offline
schedule. A deterministic algorithm is $R$-competitive if
$\ALG(I)\le R\OPT(I)$ for every input, with no additive constant.
Unless stated otherwise, we assume $g\ge2$ and positive radii: for
$g=1$, immediate singleton service is optimal, and zero-radius requests
can always be served immediately for free. Section~\ref{subsec:arbitrary-sizes}
extends the model to arbitrary indivisible request sizes.

\section{An \texorpdfstring{$O(1)$}{O(1)}-competitive algorithm for LAD}
\label{sec:line-path-load}

Capacity makes the uncapacitated maximum-delay algorithm unusable in two
different ways. A pending block may contain more than $g$ requests, and
splitting it can charge the same waiting period to several batches. Serving
each radius class separately avoids this split, but loses the waiting-cost
sharing between classes. We first isolate that loss and then give the
constant-competitive algorithm.

\subsection{Why separate radius classes lose a logarithmic factor}
\label{subsec:class-baseline}

Round each positive radius upward to a power of two. At every rounded radius
$y$, independently serve the $g$ oldest requests when $g$ are pending, and
otherwise serve all pending requests when the oldest reaches age $y$.
\begin{proposition}\label{prop:line-maximum-baselines}
For $g\ge2$, this independent-class algorithm is
$O(\log g)$-competitive.
\end{proposition}
Appendix~\ref{sec:independent-class-baseline} gives the proof. The following instance shows that the loss is
intrinsic to separating the classes.

Let $m=\lfloor\log_2(g+1)\rfloor$ and $R=2^{m-1}$. For
$i=0,\ldots,m-1$, release one request at radius $2^i$ at every time
$2j2^i$, where $j=0,\ldots,R/2^i-1$. There are
$\sum_i R/2^i=2^m-1\le g$ requests in total. The class algorithm serves
each request alone and pays $2R$ per class, hence $2mR$ altogether. Offline
serves all requests in one batch at the last arrival and pays $3R-2$.
The competitive ratio is therefore $\Omega(\log g)$ even though the entire
input fits in one batch. The lost factor comes from repeatedly paying waiting
cost that one shared batch pays once.

\subsection{Next Fit along nested time intervals}
\label{subsec:line-main-idea}

We now let radius classes share future batches. For the moment suppose that
radii are powers of two. For each radius $y$ and each multiple $b$ of $y$,
collect the requests that arrived in $[b-y,b)$. At time $b$, serve full
groups of $g$ immediately and retain the remaining requests as one
\emph{cohort}. Give this cohort the \emph{home interval} $[b,b+y)$.
These homes are nested or disjoint.

A selected home is a \emph{root} and represents one future batch, served at
its right endpoint. A cohort may join the nearest selected root containing
its home. Suppose that the candidate root has width $R$, while the cohort
has $M$ requests and home width $y$. The cohort contributes the constant
load
\begin{equation}\label{eq:line-load-overview}
  \frac{MR}{gy}=\frac{M}{g}\frac{R}{y}
\end{equation}
on its home. The first factor records capacity use; the second records how
long the candidate owner keeps that capacity. The cohort joins the root only
if the total load on its home remains at most one. Otherwise it opens a new
root. Later descendants test this new root rather than returning to older
ancestors. This is Next Fit along each path of nested homes.

The load test already implies capacity feasibility. The integral of the
load assigned to a root of width $R$ is $R/g$ per request and is at most
$R$. Hence at most $g$ requests join the root. Every assigned request has
radius at most $R$ and waits at most $2R$, so the resulting batch costs at
most $3R$.

\paragraph{A worked example.}
Let $g=8$. Two requests of radius $8$ arrive at time $0$; two requests of
radius $4$ arrive at time $4$; and at time $8$ one request of radius $2$
and one of radius $4$ arrive. Figure~\ref{fig:path-load-example} follows the
six requests.

At time $8$, the radius-$8$ cohort opens the root $[8,16)$ with load $1/4$.
The radius-$4$ cohort with home $[8,12)$ contributes $1/2$ and joins it. At
time $10$, the radius-$2$ request reveals home $[10,12)$ and would contribute
$1/2$, taking the load there to $5/4$. It therefore opens a new root and is
served at time $12$, for cost $6$. At time $12$, the final radius-$4$
request reveals home $[12,16)$. This home is disjoint from $[8,12)$, so its
test sees load $1/2$ and it joins the first root. That root serves five
requests at time $16$, for cost $24$. The algorithm pays $30$, while one
offline batch at time $8$ costs $16$.

\begin{figure}[!ht]
\centering
\begin{tikzpicture}[x=0.95cm,y=0.82cm,font=\small]
\node[anchor=west,font=\small\bfseries] at (0,6.3)
  {(a) Cohort homes and their assigned batches ($g=8$)};
\node[anchor=east] at (2.7,5.5) {root $u$};
\filldraw[fill=figblue!12,draw=figblue,thick] (3,5.2) rectangle (13,5.8);
\node at (8,5.5) {$[8,16)$: two requests of radius $8$};
\node[anchor=east] at (2.7,4.5) {accepted};
\filldraw[fill=figblue!12,draw=figblue] (3,4.2) rectangle (8,4.8);
\filldraw[fill=figblue!12,draw=figblue] (8,4.2) rectangle (13,4.8);
\node at (5.5,4.5) {two requests of radius $4$};
\node at (10.5,4.5) {one request of radius $4$};
\node[anchor=east] at (2.7,3.45) {new root $v$};
\filldraw[fill=figorange!12,draw=figorange,dashed,thick]
  (5.5,3.15) rectangle (8,3.75);
\node at (6.75,3.45) {radius $2$};
\node[anchor=west,align=left,font=\footnotesize] at (8.35,3.45)
  {At $t=10$: $\frac14+\frac12+\frac12>1$.\\Reject from $u$; open $v$.};
\draw[paper arrow] (3,2.7)--(13.6,2.7) node[right] {$t$};
\foreach \x/\t in {3/8,5.5/10,8/12,13/16}{
  \draw (\x,2.64)--(\x,2.76);
  \node[below,font=\footnotesize] at (\x,2.64) {$\t$};}
\node[anchor=west,font=\small\bfseries] at (0,1.7)
  {(b) Final load of the root $u$};
\begin{scope}[shift={(3,-1.0)}]
  \fill[figblue!16] (0,0)--(0,1.2)--(5,1.2)--(5,0.8)--(10,0.8)--(10,0)--cycle;
  \draw[figblue,thick] (0,1.2)--(5,1.2)--(5,0.8)--(10,0.8);
  \draw[dashed] (0,1.6)--(10,1.6) node[right] {$1$};
  \draw[paper arrow] (0,0)--(10.6,0) node[right] {$t$};
  \draw[paper arrow] (0,0)--(0,1.85) node[above] {load};
  \node at (2.5,0.6) {$\frac14+\frac12=\frac34$};
  \node at (7.5,0.4) {$\frac14+\frac14=\frac12$};
  \foreach \x/\t in {0/8,5/12,10/16}
    \node[below,font=\footnotesize] at (\x,0) {$\t$};
\end{scope}
\end{tikzpicture}
\[
 \text{Total area}=4\cdot\frac34+4\cdot\frac12=5
 =\frac88\cdot\text{(five assigned requests)}\le8.
\]
\caption{The six-request example from Section~\ref{subsec:line-main-idea}.
The radius-$8$ requests arrive at time $0$, the first radius-$4$ pair at
time $4$, and the other two requests at time $8$. A home starts when
its cohort becomes available for assignment, rather than at the
requests' actual arrival. The orange cohort is rejected at time $10$
and served at time $12$ for cost $6$; the five blue requests are served
at time $16$ for cost $24$. A feasible offline batch at time $8$ costs
$16$. Disjoint homes use different portions of the root's
load budget; only accepted cohorts contribute to the plotted load.}
\label{fig:path-load-example}
\end{figure}

The rejection at $[10,12)$ leaves a capacity slot unused, but it also gives
a lower-bound certificate. Either the incoming cohort contributes a
constant load, or the old load is already large throughout its home. In the
second case, the area under that old load pays for the new root. Next Fit is
important here: the first rejected descendant becomes the owner for later
descendants, so smaller roots do not repeatedly charge the same area.

\subsection{The algorithm and its guarantee}

For arbitrary radii, fix a scale $\lambda>0$ and round $x_r$ upward to the
smallest $y_r\in\{\lambda2^j:j\in\mathbb Z\}$. Define
\[
  b_r=y_r\bigl(\lfloor a_r/y_r\rfloor+1\bigr),
  \qquad H_r=[b_r,b_r+y_r).
\]
Thus the cohort with home $[b,b+y)$ is completely known at time $b$ from
requests arriving in $[b-y,b)$. Homes with the same left endpoint are
processed from larger to smaller. Full groups are served first; a nonempty
remainder becomes one cohort.

If cohort $v$ has $M_v$ requests and home $I_v$ of width $y_v$, and $u$ is
its nearest selected strict ancestor of width $R_u$, then admitting $v$
would give load
\begin{equation}\label{eq:line-admission-overview}
 L_u(v)=
 \sum_{\substack{w\text{ assigned to }u\\ I_v\subseteq I_w}}
       \frac{M_wR_u}{gy_w}
       +\frac{M_vR_u}{gy_v}.
\end{equation}
The cohort joins $u$ if $L_u(v)\le1$; otherwise it becomes a root. Every
root serves its assigned cohorts at the right endpoint of its home.
Algorithm~\ref{alg:path-load} gives the event-driven rule.

\begin{algorithm}[!ht]
\caption{Path-load aggregation for LAD}\label{alg:path-load}
\begin{algorithmic}[1]
\Require Capacity $g$; one fixed scale $\lambda>0$; deterministic tie orders
\Statex Use $\lambda=1$ for the deterministic version; for Corollary~\ref{cor:line-random-shift},
draw $U$ uniformly from $[0,1)$ once and set $\lambda=2^U$.
\State Initially no requests, cohorts, or selected roots are stored.
\Statex \textbf{On arrival of request $r$ at $a_r$:}
\State $y_r\gets\min\{\lambda2^j:\lambda2^j\ge x_r\}$;
       $b_r\gets y_r(\lfloor a_r/y_r\rfloor+1)$
\State Store $r$ in the cohort with home $[b_r,b_r+y_r)$.
\Statex \textbf{At a virtual arrival time $b$:}
\For{nonempty homes $[b,b+y)$ in decreasing $y$}
  \State $Q\gets$ the requests stored for this home
  \While{$|Q|\ge g$}
    \State Serve the first $g$ requests of $Q$ at $b$ and remove them from $Q$.\label{line:full-cohort}
  \EndWhile
  \If{$Q\ne\varnothing$}
    \State Create cohort $v$ with $I_v=[b,b+y)$, $y_v=y$, and $M_v=|Q|$.
    \State $u\gets$ the selected strict ancestor of $I_v$ of minimum width, if any
    \If{$u$ exists}
      \State $L\gets\displaystyle\sum_{\substack{w\text{ assigned to }u\\I_v\subseteq I_w}}
        \frac{M_wR_u}{gy_w}+\frac{M_vR_u}{gy_v}$
    \EndIf
    \If{$u$ exists and $L\le1$}
      \State Assign the entire cohort $v$ to $u$.\label{line:admit-cohort}
    \Else
      \State Select $v$ as a root, set $R_v=y_v$, and assign $v$ to itself.\label{line:open-root}
      \State Schedule its service at $b+y_v$.
    \EndIf
  \EndIf
\EndFor
\Statex \textbf{At the service time $b_u+R_u$ of a root $u$:}
\State Serve all requests assigned to $u$ in one batch; retire $u$.
\end{algorithmic}
\end{algorithm}

\begin{theorem}\label{thm:line-path-load}
There is a deterministic online algorithm for LAD with capacity $g$ whose
competitive ratio is at most $36+12\sqrt3<56.785$. For radii on one fixed
scaled dyadic grid, the ratio is at most
\[
  C_*:=18+6\sqrt3<28.393.
\]
The guarantees also hold for adaptively generated inputs.
\end{theorem}

\subsection{Why the analysis is additive}

The proof compares request charges inside each batch of an arbitrary
offline partition. This is necessary because one offline batch may take
parts of several online cohorts. Fix such a batch $B$ and write
\[
 X_B=\max_{r\in B}x_r,\qquad
 Y_B=\max_{r\in B}y_r,\qquad
 D_B=\max_{r\in B}a_r-\min_{r\in B}a_r.
\]
Serving $B$ at its last arrival costs $X_B+D_B$. Let $\alpha_B$ be the
fraction of its capacity occupied by requests that remain after local
preprocessing.

There are three charges. Full preprocessed batches charge at most
$4(1-\alpha_B)Y_B$ inside $B$. For each residual request assigned to a root
of width $R$, give the request weight $\beta=R/g$. This weight records the
area reserved for that request. For a threshold $\theta\in(0,1)$, call a
new root heavy when its own prospective load in its parent is at least
$\theta$; distribute the heavy root's width over the requests that created
it, giving weights $\gamma$.

The proof uses the following three bounds:
\begin{equation}\label{eq:line-certificate-summary}
 \sum_{r\in B}\beta_r\le D_B+(2+\alpha_B)Y_B,
 \qquad
 \sum_{r\in B}\gamma_r\le D_B+(2+\alpha_B/\theta)Y_B,
\end{equation}
and the total width of light roots is at most
$\sum_r\beta_r/(1-\theta)$. The first two inequalities use maximal
represented homes. They are disjoint and lie in an interval of length
$D_B+2Y_B$. Capacity pays for the remaining requests. The light-root bound
integrates the old load over the disjoint rejected homes of each parent.
Appendix~\ref{sec:line-detailed-analysis} proves these statements and also
shows why trying older roots before the nearest root can lose a logarithmic
factor in this certificate.

Each root batch costs at most three times its width. Combining the three
charges and maximizing over $\alpha_B\in[0,1]$ gives
\begin{equation}\label{eq:line-direct-comparison-overview}
 \ALG\le\sum_B\bigl(A(\theta)D_B+C(\theta)Y_B\bigr),
 \quad A(\theta)=3\left(1+\frac1{1-\theta}\right),
 \quad C(\theta)=3\left(2+\frac1\theta+\frac3{1-\theta}\right).
\end{equation}
At $\theta=1/(1+\sqrt3)$, these coefficients are
$A=6+\sqrt3$ and $C=C_*$. On a scaled dyadic grid, $Y_B=X_B$; for
arbitrary radii, $Y_B<2X_B$. This proves Theorem~\ref{thm:line-path-load}.
The threshold appears only in the proof, not in the algorithm.

\begin{corollary}\label{cor:line-random-shift}
Against an oblivious adversary, LAD with capacity $g$ admits a randomized
online algorithm with expected competitive ratio
\[
  \frac{18+6\sqrt3}{\ln2}<40.962.
\]
\end{corollary}

Draw $U$ uniformly from $[0,1)$ and use the common scale $\lambda=2^U$.
For a fixed radius $x$, its rounded value $y_U(x)$ satisfies
$\mathbb E[y_U(x)]=x/\ln2$. Since the rounding map is monotone, the largest
rounded radius in $B$ is $y_U(X_B)$. Taking expectations in
\eqref{eq:line-direct-comparison-overview} proves the corollary. The input
must be fixed independently of $U$; the exact continuous scale is interpreted
in the ideal random-real model.

\subsection{Arbitrary request sizes}
\label{subsec:arbitrary-sizes}

The same algorithm handles indivisible sizes $s_r\in[0,K]$. In each cohort,
serve requests larger than $K/2$ separately, apply Next Fit to the others,
and keep the final nonempty bin as the residual cohort. Replace $M_v$ by its
total size and $g$ by $K$ in the path load. The integral of that load bounds
the assigned size by $K$.

\begin{corollary}\label{cor:size-path-load}
The deterministic and randomized guarantees of
Theorem~\ref{thm:line-path-load} and Corollary~\ref{cor:line-random-shift}
hold for arbitrary indivisible request sizes, including zero sizes. Neither
algorithm needs a minimum-size assumption or an optimum bin-packing oracle.
\end{corollary}

Appendix~\ref{sec:arbitrary-sizes} gives the weighted preprocessing and the
complete proof.

\section{Learning-augmented capacitated LAD}
\label{sec:online-advice}

We next ask whether a good proposed schedule can improve the online cost
without sacrificing the worst-case guarantee. Advice is revealed only when
it proposes a service: at time $t$, an event specifies a capacity-feasible
set of already arrived requests to serve. The algorithm sees neither future
events nor the total advice cost. Consistency is measured against every
complete feasible advice stream, including one that is not optimal.

The difficulty is not merely deciding when to stop trusting advice. Two
schedules that have processed the same arrivals may have served different
requests. After a switch, the new schedule may refer to requests that are no
longer pending, while requests already served by the new schedule may still
be pending physically. Clearing the physical queue restores a common state,
but one advice batch may then be split across several clearances. Its
maximum waiting cost cannot be charged in full at every switch.

\begin{theorem}\label{thm:feasible-stream-monitor}
For LAD with capacity $g$, let $B$ be a deterministic $R$-competitive online
algorithm. For every $0<\eps\le1$, there is a deterministic online algorithm
such that
\begin{align}
 \ALG(I)&\le(1+\eps)\cost(\ADV),
       \label{eq:stream-consistency}\\
 \ALG(I)&\le
 \left(\frac{32}{\eps^2}+\frac{72}{\eps}+36\right)R\OPT(I).
       \label{eq:stream-robustness}
\end{align}
The consistency bound holds for every complete capacity-feasible advice
stream revealed at its service times. The robustness bound is independent
of advice quality and also holds for incomplete advice. The algorithm has
polynomial overhead for rational input and $\eps$ whenever $B$ is
polynomially simulable.
\end{theorem}

Using Theorem~\ref{thm:line-path-load} for $B$ gives
$O(\eps^{-2})$-robustness.

\subsection{Scores and switches}

The algorithm continuously simulates the advice schedule $A$ and the robust
baseline $B$. For a pending set $P$ at time $t$, let
\begin{equation}\label{eq:stream-separate-packings-overview}
 S(P)=\min_{\Pi}\sum_{G\in\Pi}\max_{r\in G}x_r,
 \qquad
 D(P,t)=\min_{\Pi}\sum_{G\in\Pi}\max_{r\in G}(t-a_r),
\end{equation}
where each minimum ranges over partitions into groups of at most $g$
requests. Sorting the requests by decreasing radius computes $S$; sorting
them by decreasing age computes $D$. The two partitions may differ, but
$S(P)+D(P,t)$ is a lower bound on the cost of clearing $P$.

For a virtual schedule $X\in\{A,B\}$, let $p_X(t)$ be the cost it has paid
and $P_X(t)$ its pending set. Its score is
\begin{equation}\label{eq:stream-source-score-overview}
 H_X(t)=p_X(t)+S(P_X(t))+D(P_X(t),t).
\end{equation}
The score is nondecreasing. If $X$ is complete and feasible, then
$H_X(t)\le\cost(X)$: restrict its future services to its current pending
requests. Thus the score is both observable online and bounded by the final
cost of a valid source.

Set
\[
 \eta=\frac{\eps}{2},
 \qquad
 \delta=\frac{\eps}{4(3+2/\eps)}.
\]
The monitor uses dyadic budget levels $b,2b,4b,\ldots$. At level $b$, advice
is eligible when $H_A<b$, while the baseline is eligible when
$H_B<\delta b$. The monitor prefers advice whenever it is eligible. If
neither source is eligible, it doubles the level. A change of level without
a change of source requires no physical action.

When the active source changes, the algorithm first serves every physically
pending request. It then follows services of the new source, restricted to
requests that remain physically pending. The cleanup on a return to advice
also removes requests that virtual advice served while the baseline was
active. Consequently, during an $X$ block the physical pending set is a
subset of $P_X$. Appendix~\ref{sec:advice-full} gives the complete event
order and pseudocode, including simultaneous arrivals, services, and score
timers.

Two simple packing procedures provide the required cleanups. An ordinary
cleanup costs at most
\begin{equation}\label{eq:stream-ordinary-clearance-overview}
 2\bigl(S(P)+D(P,t)\bigr).
\end{equation}
For any $\eta>0$, a biased cleanup costs at most
\begin{equation}\label{eq:stream-biased-clearance-overview}
 (1+\eta)D(P,t)+(1+1/\eta)S(P).
\end{equation}
The monitor uses the biased cleanup when leaving advice, because consistency
has only an $\eps$ margin for waiting cost. It uses the ordinary cleanup when
returning to advice. The smaller baseline budget pays for the latter and for
the service-cost part of leaving advice.

\subsection{The residual-advice charge}

The central accounting step pays the waiting part of every exit from advice
without reusing an advice batch. Fix the original advice batches
$(G_j,T_j)$, and let $Z$ be the set of requests already served physically.
Define the remaining advice cost
\begin{equation}\label{eq:stream-residual-potential-overview}
 \Phi(Z)=\sum_j c(G_j\setminus Z,T_j),
 \qquad c(\varnothing,T_j)=0.
\end{equation}
This is an analysis potential, not future information available to the
algorithm. It starts at $\cost(A)$ and never increases.

A physical service that follows advice decreases $\Phi$ by its own cost.
Suppose instead that the monitor leaves advice at time $t$ and clears a
pending set $P$. Every request in $P$ is still pending in virtual advice.
For each affected advice batch, deleting these requests either removes the
batch or advances its oldest surviving arrival to at least $t$. The decrease
in its waiting cost is therefore at least the maximum age of the deleted
requests from that batch. Since the affected advice batches form a feasible
partition of $P$, the total decrease is at least $D(P,t)$. Hence
\begin{equation}\label{eq:stream-residual-credit-overview}
 F_A+\sum_{e:A\to B}D(P_e,t_e)\le\cost(A),
\end{equation}
where $F_A$ is the cost of physical services that follow advice. Successive
switches use successive decreases of $\Phi$; none charges the original
maximum waiting cost twice.

\subsection{Completing the two bounds}

The remaining accounting is geometric. An advice block ending at level $b$
pays at most $(1+\eta)b$ for its followed services and waiting cleanup. A
baseline block ending at level $b$ pays at most $2\delta b$, including its
exit cleanup. Each level is the last level of at most one block of each
source. If advice is complete, the levels used by baseline blocks and by
exits from advice each sum to at most $2\cost(A)$. Together with
\eqref{eq:stream-residual-credit-overview}, this gives
\[
 \ALG\le
 \bigl[1+\eta+2\delta(3+1/\eta)\bigr]\cost(A)
 =(1+\eps)\cost(A).
\]

For robustness, no bound on the advice cost is used. The largest opened
level is at most $2\cost(B)/\delta$, and all opened levels sum to less than
twice that value. Applying the same block bounds gives
\[
 \ALG\le
 4\left[\frac{1+\eta}{\delta}+3+\frac1\eta\right]\cost(B)
 =\left(\frac{32}{\eps^2}+\frac{72}{\eps}+36\right)\cost(B).
\]
Since $\cost(B)\le R\OPT$, Theorem~\ref{thm:feasible-stream-monitor}
follows. Appendices~\ref{sec:advice-full} and~\ref{sec:advice-analysis}
prove the cleanup, block, and completion statements used here.

The dependence on $\eps$ cannot be removed. Already at one location, a
deterministic $(1+\eps)$-consistent algorithm that is $R$-robust on inputs
of at most $g$ requests must satisfy
\[
 R\ge1+\frac{g-1}{1+\eps(g-1)}.
\]
Thus any robustness bound uniform in $g$ is at least $1+1/\eps$; the proof
appears in Appendix~\ref{sec:advice-analysis}.

\newcommand{\ExperimentInputs}{180}
\newcommand{\ExperimentRuns}{9180}
\newcommand{\ExperimentBaseline}{2.69}
\newcommand{\ExperimentPoorQuarter}{602.72}
\newcommand{\ExperimentPoorHalf}{223.62}
\newcommand{\ExperimentPoorOne}{75.42}

\section{Controlled Advice Experiments}\label{sec:experiments}
We evaluate the stated monitor, including its unmodified constants, on
60 synthetic request instances of size $12$: 20 fixed seeds for each of
bursty, spread-out, and mixed-scale arrivals. Testing capacities
$g\in\{2,4,8\}$ gives \ExperimentInputs\ cases with exact offline optima.
An exhaustive subset-partition DP computes OPT; the experiment driver reveals
advice services individually at their proposed times. The monitor never
receives future advice. We use $\eps\in\{1/4,1/2,1\}$ and compare with
path-load without advice and with directly following feasible advice.

\paragraph{Protocol and measurement.}
The three families vary arrival density and the range of request radii;
the same request instance is reused across capacities and monitor settings.
Every input and outcome is retained. We report the mean of per-instance
ratios, not a ratio of pooled costs. Pointwise 95\% intervals use 2000
stratified bootstrap replicates over instance seeds, keeping the capacity
variants of each instance together. The benchmark oracle partitions subsets
into batches of size at most $g$, each served at its last arrival. Its
exponential running time confines this exact-OPT study to small inputs.
It is independent of the polynomially simulable online rule. All monitor
events and costs use rational arithmetic; a separate evaluator checks the
physical schedule for complete, capacity-feasible service and recomputes
its total cost.

\paragraph{Systematic lateness.}
Starting from an optimal schedule with $m$ batches, shift each proposed
service later by $\rho\OPT/m$. This preserves its partition and makes
$\ADV/\OPT=1+\rho$ exactly, allowing direct comparison across capacities.
The experiment varies this ratio from $1$ to $10001$ to expose both the
consistency and fallback regimes. The large errors are deliberate stress
tests, rather than an assumed prediction-error distribution.

Figure~\ref{fig:advice-experiments} shows mean finite-input cost divided by
exact OPT. With optimal advice, every monitor attains ratio $1.00$, versus
\ExperimentBaseline\ for the baseline on average. With systematic lateness,
the monitor initially follows advice: all three settings still have ratio
$11$ when $\ADV/\OPT=11$. Its proven thresholds are conservative.
At quality $1001$, the mean ratios are \ExperimentPoorQuarter,
\ExperimentPoorHalf, and \ExperimentPoorOne\ for $\eps=1/4,1/2,1$,
respectively, all above the baseline. Smaller $\eps$ protects consistency
but postpones fallback. The observed leveling at the largest tested errors
is consistent with fallback; the universal bound comes from the proof.

At the same level $b$, baseline eligibility uses only $b/176$, $b/56$,
and $b/20$ in these settings; advice uses $b$.

\begin{figure}[htbp]
\centering
\includegraphics[width=\textwidth]{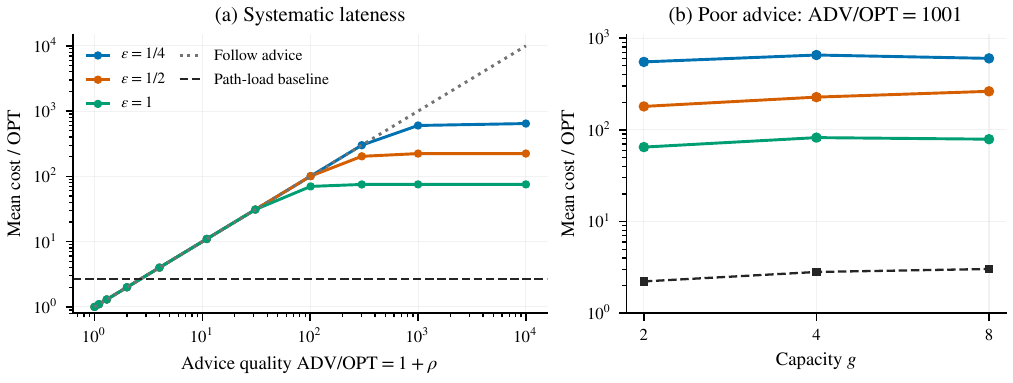}
\caption{Controlled advice evaluation, using all generated instances.
(a) Shift every optimal advice batch by $\rho\OPT/m$, where $m$ is its number
of batches, so $\ADV/\OPT=1+\rho$ exactly.
(b) Capacity dependence at advice quality $1001$.
Shading and error bars are pointwise 95\% stratified cluster-bootstrap
intervals, resampling request-instance seeds within each family and retaining
their capacity variants together. Colors identify the same $\eps$ in both
panels; the dashed black curve is the path-load baseline.}
\label{fig:advice-experiments}
\end{figure}

\paragraph{Feasible partition errors and scope.}
A separate treatment selects optimal batches, shuffles their requests, and
repartitions them into capacity-feasible groups, serving each at its last
arrival. The realized advice-to-OPT ratios range from $1$ to below $4.053$.
In this range all three monitors have the same physical cost as following
advice and make no nonempty switch. This complements the lateness test,
but the two treatments do not isolate error structure at matched quality.
The tests establish neither
average-case superiority over path-load nor scalability on large demand
traces. They isolate the behavior of the stated, untuned monitor and
identify switching thresholds as a concrete target for improvement.

Appendix~\ref{app:advice-experiments} gives the full protocol, family-wise
curves, a separate capacity-feasible repartitioning experiment, and all
\ExperimentRuns\ exact-rational monitor runs. These controlled schedules
isolate advice quality; they are not predictions produced by a trained model.

\FloatBarrier

\section{Conclusion and future directions}
\label{sec:conclusion}

The constant competitive ratio for capacitated LAD leaves several
questions open. What are the optimal deterministic and randomized
competitive ratios? Our bounds show that capacity does not force a
logarithmic loss, but a gap remains even against the lower bounds from
a single location. For online advice, our robustness is
$O(\eps^{-2})$, while the deterministic lower bound is
$\Omega(\eps^{-1})$ uniformly in capacity. Closing this gap is another
natural question.
The controlled advice experiments also expose a practical limitation of
the current constants: middle-error inputs can trigger fallback only
after paying much more than the baseline. Improving these thresholds
while retaining consistency and robustness is a further direction.

Extending the constant guarantee beyond a line is a further direction.
Does capacitated JRP with maximum delay admit a constant competitive
ratio or a constant approximation ratio? For rooted-tree MLA, how
should the competitive ratio depend on capacity and tree depth?
The LAD admission rule uses nested intervals to control both capacity
and the cost of opening batches. For branching service costs, requests
can share some edges while paying separately for others; an extension
must account for both forms of cost sharing.

Capacity constraints also merit study in online service and
facility-location problems with waiting costs. How does maximum delay
interact with capacity that can be reused through an opened facility,
rather than assigned to a single instantaneous service?

\section*{Declaration of AI Use}
Generative AI tools assisted mathematical exploration, proof development,
drafting, literature searches, experimental-code preparation, and LaTeX
preparation. The authors are responsible for the results, references,
code, and final text.
\appendix
\section{Capacitated TCP acknowledgment with maximum delay}
\label{sec:online}
The TCP algorithm below handles the single-location special case of LAD
and each distance class in Section~\ref{subsec:class-baseline}.
We prove its optimal deterministic competitive ratio of two, followed
by the randomized lower bound stated in the introduction.
\subsection{A 2-competitive algorithm}

For a fixed setup cost $w>0$, a batch costs $w$ plus its maximum age.
The algorithm keeps pending requests in arrival order. Whenever at
least $g$ are pending, it serves the $g$ oldest and repeats if necessary.
Otherwise, it serves the entire pending set when its oldest request
reaches age $w$.  A full-batch trigger
has priority over a simultaneous timer.  Every resulting batch costs at
most $2w$.

Algorithm~\ref{alg:fixed-location} records this rule. All services occur at the
current event time. A stale timer is replaced whenever its buffer changes.
\begin{algorithm}[!ht]
\caption{The fixed-location full-or-age rule}\label{alg:fixed-location}
\begin{algorithmic}[1]
\Function{FullOrAge}{$g,w$}
 \State Maintain an arrival-ordered queue $Q$.
 \State On an arrival block, append its requests to $Q$ in a fixed tie order.
 \While{$|Q|\ge g$}
  \State Serve and remove the first $g$ requests.
 \EndWhile
 \State After arrival actions, if $Q\ne\varnothing$ and its oldest age is at least $w$,
        serve all of $Q$ and empty it.
 \State Otherwise schedule the nonempty queue's timer at its oldest arrival plus $w$.
\EndFunction
\State A timer serves the entire nonempty queue. If $w=0$, serve arrivals immediately.
\end{algorithmic}
\end{algorithm}

\begin{lemma}[Greedy short prefixes]
\label{lem:short-prefixes}
For requests at one location, call a batch short if it has at most $g$
members and arrival span at most $w$.  The greedy partition into maximal arrival-ordered short prefixes
minimizes the number of short batches, even when nonconsecutive groups
are allowed.
\end{lemma}
\begin{proof}
Let $a$ denote the earliest arrival, and let $G$ denote the first
$q=\min\{g,|\{r:a_r\le a+w\}|\}$ requests in arrival order.  In any
short partition, the batch $Q$ containing the first request has
$|Q|\le q$ and all its arrivals in $[a,a+w]$.

For each request $z\in Q\setminus G$, we choose a distinct request
$y\in G\setminus Q$ and swap $z$ into the batch containing $y$.
There are enough such $y$ because $|Q|\le|G|$.  Requests of $G$ precede
all displaced requests in arrival order.  Thus an inserted $z$ arrives
no earlier than the removed $y$, and no later than $a+w$.  In a receiving
batch with original earliest arrival $b\ge a$, the new earliest arrival
does not decrease, and every inserted arrival is at most $a+w\le b+w$.
Its span remains at most $w$ and its cardinality does not increase.
We move any remaining requests of $G\setminus Q$ into $Q$, deleting them
from their old batches.  The batch $Q$ has become exactly $G$, and no
additional batch has been created. Empty batches are deleted.

We remove $G$ and repeat this exchange on the remaining instance. Induction
shows that the greedy partition uses no more batches than the original
one.
\end{proof}

\begin{proposition}[TCP upper bound]
\label{prop:fixed-unit-two}
For capacity $g$, full-or-age-$w$ is $2$-competitive for TCP.
\end{proposition}
\begin{proof}
Its batches are exactly the maximal short prefixes: a prefix ends when
it reaches cardinality $g$ or when the first request's age reaches $w$.
Boundary arrivals are included before a timer is processed. Let $G$
denote the number of algorithm batches.

We consider any offline batch of at most $g$ requests, with arrival
span $\Delta$. Its members are divided into greedy prefixes of span at most
$w$, without a cardinality restriction; each piece still has at most
$g$ members.  If there are $h\ge2$ pieces, consecutive first arrivals
differ by more than $w$, and hence $(h-1)w<\Delta$.  For $h=1$ the bound
$hw\le w+\Delta$ is immediate.  Thus its pieces satisfy
$hw\le w+\Delta$, no larger than that comparison batch's actual cost.
Applying this to all comparison batches gives a feasible short partition
with number $M$ satisfying $wM\le\cost(\mathrm{SOL})$.
Lemma~\ref{lem:short-prefixes} gives $G\le M$.  Hence
$\ALG\le2wG\le2\cost(\mathrm{SOL})$.
\end{proof}

\subsection{A tight deterministic lower bound}
\begin{theorem}\label{thm:fixed-k-two}
For every fixed $g\ge2$, every deterministic algorithm for capacitated
TCP acknowledgment with maximum delay has competitive ratio at least two.
\end{theorem}
\begin{proof}
We fix setup cost $c>0$ and an even integer $N$. The adversary releases
the first request at $a_1=0$. After the algorithm serves request $i$
at $t_i$, request $i+1$ arrives at $a_{i+1}=t_i+\delta_i$, where
$\delta_i>0$. If the algorithm never serves the unique pending request,
the adversary ends the input. Otherwise, we write $x_i=t_i-a_i$ and
$\Delta=\sum_{i<N}\delta_i$. Every online batch is a
singleton, and hence
\[
 \ALG\ge Nc+\sum_{i=1}^N x_i.
\]
We compare the algorithm with two offline schedules. The first serves the pairs
$(1,2),(3,4),\ldots,(N-1,N)$ at their second arrivals. The second serves
requests $1,N$ singly and serves $(2,3),(4,5),\ldots,(N-2,N-1)$ at their
second arrivals. Their costs $O_1,O_2$ satisfy
\[
 O_1+O_2=(N+1)c+\sum_{i<N}x_i+\Delta.
\]
Indeed every adjacent arrival gap is charged exactly once, and a pair's
additive waiting and maximum waiting are both that gap. Both schedules use at most two requests per batch. With
$X=\sum_{i<N}x_i$,
\[
 \frac{\ALG}{\OPT}
 \ge \frac{2(Nc+X+x_N)}{(N+1)c+X+\Delta}
 \ge \frac{2Nc}{(N+1)c+\Delta}.
\]
For fixed $N$, make $\Delta/c$ arbitrarily small, and then let even $N$
grow. The resulting ratios approach two. For a deterministic algorithm, these release times can be fixed
before presenting the input.
\end{proof}

\subsection{An oblivious randomized lower bound}
\label{sec:randomized-fixed-location}
The fixed-location special case also gives a randomized obstruction for
capacitated LAD. The following lower bound concerns
pure multiplicative guarantees, as in Section~\ref{sec:model}.

\begin{theorem}\label{thm:randomized-fixed-lower}
For cardinality capacity $g$, every randomized online algorithm has ratio
at least
\[
 \frac{g}{1+(g-1)(1-e^{-1})}
\]
against an oblivious adversary, already at one location. Consequently a
competitive constant uniform in $g$ is at least $e/(e-1)$. The same
constant lower-bounds the arbitrary-size model.
\end{theorem}
\begin{proof}
We normalize the setup cost to one. When all requests fit together, any
algorithm can be normalized causally to clear all pending requests when
its simulated original execution serves the oldest pending request.
This oldest request has not previously been served by the simulation:
otherwise a normalized clearance would already have removed it. The new
batch has no larger maximum age than that distinct simulated batch, and
costs no more. Other simulated actions can be ignored. This also
normalizes offline schedules to consecutive groups.

We write $\delta_i=a_{i+1}-a_i$ for successive arrival gaps. A normalized
algorithm cuts a gap at a timeout $z_i$ if no arrival intervenes. A cut
contributes $1+z_i$ to its total cost; an uncut gap contributes $\delta_i$.
The final contribution is $1+z_n$. The setup in the final contribution
and the setups at cuts count exactly the physical batches, while the
remaining terms partition their waiting times. Offline chooses each cut
independently, obtaining
\begin{equation}\label{eq:nonbinding-gap-opt}
 \OPT=1+\sum_{i<n}\min\{\delta_i,1\}.
\end{equation}

For the lower bound, the input has $n=g$ requests, starts at zero,
and has $n-1$ independent rate-one exponential arrival gaps. This input
distribution is independent of the algorithm's randomness. Conditional
on the past, its prospective no-arrival timeout $z$ is independent of
the next gap $G$. For every $z\in[0,\infty]$ the expected gap contribution
is
\[
 \E[\min\{G,z\}]+\Pr(G>z)=1.
\]
The final contribution is at least one, and hence $\E[\ALG]\ge n$, whereas
\eqref{eq:nonbinding-gap-opt} gives
$\E[\OPT]=1+(n-1)(1-e^{-1})$.
A guarantee on every fixed input must also hold after averaging over
this independent distribution, proving the displayed ratio. Letting $g$ grow gives the uniform bound.
For arbitrary sizes, $n$ requests of size $1/n$ give the same bound
as $n$ grows.\qedhere
\end{proof}

\section{Analysis of the independent-class baseline}
\label{sec:independent-class-baseline}
Section~\ref{subsec:class-baseline} defines the algorithm and gives its
matching lower-bound instance. Here we prove the upper bound.

\begin{proof}[Proof of Proposition~\ref{prop:line-maximum-baselines}]
Fix an offline batch with maximum rounded radius $R$ and arrival span
$\Delta$. If it contains $n_x$ requests of radius $x$, its restriction can
be served together at its last arrival for at most $x+\Delta$, or as
singletons at their arrivals for cost $n_xx$.  Each restriction fits in one batch, giving total cost at most
\[
 \sum_{x\le R}\min\{x+\Delta,n_xx\}
 \le\sum_{\substack{x\le R\\x\text{ dyadic}}}
             \min\{x+\Delta,gx\}.
\]
If $\Delta=0$, this sum is at most $2R$.  Otherwise the radii
$x\le\Delta/g$ contribute at most $2\Delta$ by geometric summation.
There are at most $\lceil\log_2g\rceil+1$ dyadic radii in
$(\Delta/g,\Delta]$, each contributing at most $2\Delta$. Radii
$x>\Delta$ contribute at most $\sum_{x\le R}2x\le4R$.
The sum is at most $(2\lceil\log_2g\rceil+4)(R+\Delta)$.
Applying these restrictions to all offline batches gives a feasible
comparison schedule for each class. The factor-two TCP guarantee from
Proposition~\ref{prop:fixed-unit-two} and the factor-two radius rounding
therefore give the claimed competitive ratio.
\end{proof}

\section{Proof of capacity transfer for additive delay and deadlines}
\label{sec:review}
We derive the capacity-transfer bounds discussed in
Section~\ref{sec:related-work} for additive waiting costs and deadlines.
For a rooted-tree batch $B$, let $p_r$ be the length of request $r$'s
root path. Order requests by a depth-first traversal and cut this order
at a uniformly chosen residue modulo $g$. An edge used by $n_e>0$
requests appears in an expected $1+(n_e-1)/g$ pieces. Summing over edges
bounds the service cost by
\[
 (1-1/g)f(B)+\frac1g\sum_{r\in B}p_r.
\]
For sizes normalized to capacity one, serve requests larger than $1/2$
separately. Group the others by their cumulative start positions in
shifted intervals of length $1/2$. Each group fits, and an edge is used
by at most $1+2\sum_{r\in B_e}s_r$ groups in expectation, where $B_e$
contains its small requests. Large requests cost at most twice their
volume charge. The resulting bound is $f(B)+2\sum_{r\in B}s_rp_r$.
Choosing the cheapest offset makes both constructions deterministic.
Only polynomially many candidates are needed: an already feasible
batch need not be cut, and otherwise there are $g<|B|$ unit offsets.
For arbitrary sizes, the partition changes only when a shifted boundary
crosses a cumulative start position, giving at most $|B|$ breakpoints.

Let $S_A,D_A$ be the service and additive waiting costs of the virtual
schedule. Set $\alpha=1-1/g$ and write $V=\sum_r p_r/g$ for capacity
$g$, or $V=\sum_r s_rp_r$ for sizes. Every feasible batch pays at least
its volume charge, so $V\le S^*$ for any capacitated comparator with
service cost $S^*$ and waiting cost $D^*$. Keeping all service times
therefore gives
\[
 \ALG\le\alpha S_A+D_A+V
 \quad\text{or}\quad
 \ALG\le S_A+D_A+2V,
\]
respectively. An uncapacitated ratio $R$ immediately yields competitive ratios
$R+1$ and $R+2$, respectively. With deadlines, $D_A=D^*=0$, giving
$(\alpha R+1)S^*$ for capacity $g$. For additive waiting, run the baseline
with service costs scaled by $\alpha$, whenever it supports this scaling:
\[
 \ALG\le R(\alpha S^*+D^*)+S^*
 \le\max\{R,\alpha R+1\}(S^*+D^*).
\]
All partitions use only the current virtual batch. A capacity-feasible advised schedule is also a valid comparator, so
these inequalities transfer both consistency and robustness.

\def\lineappendixmode{1}
\ifdefined\lineappendixmode
\section{Detailed analysis of the path-load algorithm}
\label{sec:line-detailed-analysis}
This appendix proves the feasibility and charging statements summarized in
Section~\ref{sec:line-path-load}. It also explains why the nearest-owner
rule is needed by the batchwise certificate.
\else
\section{An \texorpdfstring{$O(1)$}{O(1)}-competitive algorithm for LAD}
\label{sec:line-path-load}

In the uncapacitated algorithm of Lu et al.~\cite{lu2026maximum},
a dynamic program sets a deadline for the current pending block, which
is then served in one batch. Its analysis uses the fact that batches
with overlapping arrival spans can be merged without increasing cost.
Capacity changes both steps: a pending block may contain more than $g$
requests, and merging may be infeasible. Splitting a block at its deadline
can repeat its maximum delay across several batches; serving only part
of it leaves old requests for later services. Neither modification
inherits the original cost bound.

Before developing the new algorithm, we examine the loss incurred by
serving distance classes independently. The instance below shows that
capacity alone does not explain this loss: all its requests fit in one
batch, but the algorithm fails to share their waiting costs.

\subsection{Why independent distance classes lose a logarithmic factor}
\label{subsec:class-baseline}

The dyadic-class algorithm first rounds each positive radius upward to a
power of two. At each rounded radius $x$, it independently runs the
full-or-age algorithm with setup $x$: whenever at least $g$ requests
are pending, it serves the $g$ oldest; otherwise, it serves all pending
requests when the oldest reaches age $x$. Full batches take priority
over the age rule; classes never share a service.

Algorithm~\ref{alg:independent} states this baseline.
\begin{algorithm}[!ht]
\caption{The independent-class LAD baseline}\label{alg:independent}
\begin{algorithmic}[1]
\Require Capacity $g$
\State Serve zero-radius requests immediately for free.
\State Assign every other request $r$ to its rounded radius $y_r=2^{\lceil\log_2x_r\rceil}$.
\State At each rounded radius $y$, run \Call{FullOrAge}{$g,y$} from Algorithm~\ref{alg:fixed-location}.
\State Execute the resulting services independently; do not merge groups.
\end{algorithmic}
\end{algorithm}

\begin{proposition}\label{prop:line-maximum-baselines}
For $g\ge2$, the independent-class algorithm is $O(\log g)$-competitive.
\end{proposition}
The proof of the upper bound appears in
Appendix~\ref{sec:independent-class-baseline}.

\paragraph{A matching lower-bound instance.}
Let $m=\lfloor\log_2(g+1)\rfloor$ and $R=2^{m-1}$. For each
$i=0,\ldots,m-1$, set $x_i=2^i$ and release one request at radius
$x_i$ at each time $2jx_i$, where $j=0,\ldots,R/x_i-1$.
The total number of requests is
\[
 \sum_{i=0}^{m-1}\frac{R}{x_i}=2^m-1\le g.
\]
At radius $x_i$, every request reaches age $x_i$ before the next
request in its class arrives. The algorithm therefore serves it alone
at time $(2j+1)x_i$, paying $2x_i$. Each class contributes $2R$, so
$\ALG=2mR$.

Offline can serve all requests in one batch at the last arrival time
$2R-2$. This batch respects capacity, has maximum radius $R$ and
maximum waiting time $2R-2$, and costs $3R-2$. Consequently,
\[
 \frac{\ALG}{\OPT}\ge\frac{2mR}{3R-2}\ge\frac{2m}{3}
 =\Omega(\log g).
\]
Thus the logarithmic bound is tight for this algorithm. The loss occurs
even though all requests in the entire input fit in one batch:
independent class timers repeatedly pay waiting costs that a shared
batch pays only once.

\subsection{The main idea: Next Fit along nested intervals}
\label{subsec:line-main-idea}

We therefore control which groups can join a future batch while they
wait. The challenge is to retain sharing across distances without making
small groups wait too long for a distant request. Counting requests
alone ensures that a batch fits, but does not pay for groups that must
be served separately. Our \emph{path load} accounts for both group size
and waiting time scale. A failed admission test opens a new batch and
provides a charge for its cost. Later groups inside that batch's interval
try to join it, rather than returning to its ancestors: this is the
Next Fit idea along nested intervals. We explain the rule and its
analysis through an example before giving the formal proofs.

\paragraph{Giving a group a time scale.}
For this explanation, suppose radii are powers of two. For each radius
$y$, collect the requests arriving in $[b-y,b)$, where $b$ is a multiple
of $y$. At time $b$, this group is completely known. Serve full batches
of $g$ requests immediately and retain the remainder as a single group,
called a \emph{cohort}. A full batch costs at most $2y$: its service
cost is $y$ and its maximum delay is at most $y$. These batches are
paid for by the volume lower bound
\[
 \OPT\ge \frac1g\sum_r \text{radius}(r).
\]
Indeed, a batch's largest radius is at least the sum of its requests'
radii divided by $g$.

Give the remaining cohort the \emph{home interval} $[b,b+y)$. It could
be served at the right endpoint for cost at most $3y$. The home
intervals are nested or disjoint, which provides a way to combine
different radii: a cohort may instead join the future batch of a
containing interval. We call an interval selected to open a batch a
\emph{root}, and call that root the \emph{owner} of the cohorts assigned
to it. Its batch is served at its right endpoint. Larger intervals
are processed first when their left endpoints coincide. Thus a
containing root is already known when a cohort decides whether to join
it. Assignments and service times are never revised.

\paragraph{Why count load rather than requests?}
Suppose a root has interval length $R$. A cohort of $M$ requests with
home length $y$ reserves $MR/g$ units of area under this interval.
Spread this area uniformly over its own home. Its load there is
\[
 \frac{MR}{gy}=\frac{M}{g}\cdot\frac{R}{y}.
\]
The first factor is the fraction of capacity used; the second compares
the owner's time scale with the cohort's own time scale. A short home
therefore contributes more load than a long home with the same number
of requests. We allow the assignment only if the total load stays at
most one throughout the new home. The total area is then at most $R$.
Since every assigned request contributes area $R/g$, at most $g$
requests can be assigned to this root.

When the test fails, the cohort opens its own root. Later cohorts
inside this new interval test admission to the new root, rather than
returning to its ancestors. This is the Next Fit idea along each path
of nested intervals. Disjoint subintervals can still use the earlier
root. This distinction will be essential when we account for the cost
of rejected cohorts.
Section~\ref{subsec:line-first-fit} shows that trying older roots first
can destroy the certificate used to compare these costs with offline batches.

\paragraph{A worked example.}
Let $g=8$. Two requests of radius $8$ arrive at time $0$; two of radius
$4$ arrive at time $4$; and at time $8$ one request of radius $2$ and
one of radius $4$ arrive. Figure~\ref{fig:path-load-example} follows all
six requests. No full batch is removed in preprocessing.

At time $8$, the radius-$8$ cohort opens a root with home $[8,16)$,
and its two requests contribute load $2/8=1/4$. The radius-$4$ cohort
has home $[8,12)$ and adds load $2\cdot8/(8\cdot4)=1/2$ there.
It is accepted. At time $10$, the radius-$2$ request reveals its home
$[10,12)$. Its prospective load is $8/(8\cdot2)=1/2$, taking the
total on this home to $5/4$. It is rejected even though a fifth request
would fit in the batch. Instead, it opens a root and is served at
time $12$, for cost $2+(12-8)=6$.

At time $12$, the last radius-$4$ request reveals home $[12,16)$.
This home is disjoint from $[8,12)$, and its load test is only
$1/4+1/4=1/2$. It joins the first root. The resulting five-request
batch is served at time $16$, for cost $8+(16-0)=24$. The algorithm's
total cost is therefore $30$.

Offline can serve all six requests at time $8$, for cost $8+8=16$.
Compared with this feasible solution, the algorithm pays more to fix
its groups and service times online. Leaving a capacity slot unused
preserves a load budget for subsequent groups.

\paragraph{What pays for opening another batch?}
The example has load $3/4$ already present throughout the rejected
home $[10,12)$. The area under this old load is $3/2$, enough to pay
its length $2$ within a constant factor. More generally, a failed test
means that either the incoming cohort contributes at least $1/2$, or
the old load exceeds $1/2$ on its home. In the first case, charge the
new root's length to the requests that created it. In the second,
charge it to twice the old load integrated over its home. A root's
actual batch cost is at most three times its length, so it suffices
to pay for these lengths.

The second charge must not reuse the same area at every smaller scale.
This is exactly why we use Next Fit: the first rejected home becomes
responsible for later rejections inside it. The homes that charge a
fixed parent root are disjoint. A request assigned to that parent
therefore pays at most twice its reserved area in total, even if many
smaller roots are subsequently opened. Descendants of a child use the
child's assigned requests, not the parent's requests.

\paragraph{Why are these charges lower bounds on the offline cost?}
We first charge online roots to requests and then bound these charges
\emph{inside each feasible offline batch}. The second comparison is
necessary because a batch may combine requests from different roots
without taking whole cohorts.

Fix an offline batch whose largest radius is $Y$ and whose arrival
span is $D$. Its cost is at least $Y+D$. Requests with owners of length
at most $Y$ have small reserved-area charges, whose sum is controlled
by the batch's capacity. For larger owners, Next Fit gives a
\emph{unique-owner property}: inside each maximal home represented in
the batch, all requests owned above that home share the same owner.
Its load is at most one, so their area is bounded by the home length.
These maximal homes are disjoint and lie in a time interval of length
$D+2Y$.

The own-cohort charges admit a similar division: maximal roots are
controlled by disjoint intervals, and the remaining per-request
charges by capacity. Both certificates are therefore bounded by a
constant times the offline batch cost. Summing over its partition
gives the desired comparison with $\OPT$. The formal analysis below
establishes these bounds, optimizes the threshold, and accounts for
preprocessing and radius rounding.

\begin{theorem}\label{thm:line-path-load}
There is a deterministic online algorithm for LAD with
capacity $g$ whose competitive
ratio is at most $36+12\sqrt3<56.785$. For radii on a fixed scaled dyadic grid, the ratio is at most
\[
 C_*:=18+6\sqrt3<28.393.
\]
These deterministic guarantees hold also for adaptively generated inputs.
\end{theorem}

\fi

\subsection{Strict time rounding and local preprocessing}

We write $a_r$ for the arrival time of request $r$ and $x_r>0$ for its
radius. Each batch contains at most $g$ requests.
For a fixed $\lambda>0$, the radius grid is
$\{\lambda2^j:j\in\mathbb Z\}$. The deterministic algorithm uses
$\lambda=1$; we keep the scale explicit for the randomized extension.
The algorithm rounds $x_r$ upward to the smallest grid value
$y_r\ge x_r$ and gives $r$ the virtual arrival
\[
 b_r=y_r\bigl(\lfloor a_r/y_r\rfloor+1\bigr).
\]
Thus $0<b_r-a_r\le y_r$. Its \emph{home interval} is
$H_r=[b_r,b_r+y_r)$. Any two home intervals are disjoint or one contains the other.
At time $b$, the full cohort with home $[b,b+y)$ consists of the requests
of rounded radius $y$ that actually arrived in $[b-y,b)$. The cohort is
therefore known from the strict past. An arrival at time $b$ cannot change
it, even when arrivals depend on earlier service decisions.

At a cohort's virtual arrival, the algorithm serves full groups of
$g$ requests in a fixed deterministic order. It retains the remaining
requests as one \emph{residual cohort}, if any remain. Thus every
preprocessed batch is full, and each residual cohort fits in one batch
(line~\ref{line:full-cohort} of Algorithm~\ref{alg:path-load}).

Every preprocessed batch at a home of width $y$ has actual waiting at
most $y$ and radius at most $y$. Its cost is at most $2y$, and hence at
most $4y/g$ times its cardinality. If $P$ denotes the set of preprocessed
requests, then
\begin{equation}\label{eq:line-prepaid}
 \ALG_{\mathrm{pre}}\le \frac4g\sum_{r\in P}y_r.
\end{equation}
We use the loose factor four so that the same accounting also applies
to arbitrary sizes in Appendix~\ref{sec:arbitrary-sizes}.
The bound assigns preprocessing cost to individual requests. Later we
combine these charges with those of residual requests inside each
offline batch, using their shared capacity constraint.

\subsection{The path-load rule}

We write $v$ for a nonempty residual cohort, $I_v=[b_v,b_v+y_v)$
for its home, and $M_v$ for its cardinality. The algorithm processes
homes with the same left endpoint in decreasing order of width. Some homes are selected as
\emph{roots}, each representing one future batch. Every cohort is
assigned as a whole to a selected root containing its home. We call
that root the cohort's \emph{owner}; these interval roots are distinct
from the half-line root at the origin.

When cohort $v$ arrives, let $u$ denote its nearest selected strict
ancestor, if one exists, and write $R_u=y_u$. To decide whether $v$ can
join $u$, the algorithm distributes a load $M_vR_u/(gy_v)$ over $I_v$.
Its integral is $M_vR_u/g$, the interval length reserved for these
requests. Since earlier homes meeting $I_v$ contain
it, their contributions are constant throughout $I_v$. Admitting $v$ would give total load
\begin{equation}\label{eq:line-admission}
 L_u(v)=
 \sum_{\substack{w\text{ already assigned to }u\\ I_v\subseteq I_w}}
       \frac{M_wR_u}{gy_w}
       +\frac{M_vR_u}{gy_v}.
\end{equation}
If $u$ exists and $L_u(v)\le1$, the algorithm assigns $v$ to $u$.
Otherwise, $v$ becomes a new root and owns its own cohort
(lines~\ref{line:admit-cohort}--\ref{line:open-root} of
Algorithm~\ref{alg:path-load}). A root $u$ serves all its assigned
requests at time $b_u+R_u$. This service time and all earlier assignments
remain fixed.

All information used in~\eqref{eq:line-admission} is available online.
An ancestor home is revealed before its descendants, with the specified
order handling equal left endpoints. A newly revealed descendant starts
strictly before its ancestor's right endpoint, and hence its selected ancestor
has not served. Homes starting at an endpoint are outside the interval
ending there; simultaneous arrivals and due services cause no ambiguity.

All tie choices are fixed in advance. Empty cohorts create no events.

\begin{lemma}\label{lem:line-load-capacity}
Every root batch is capacity feasible, and its actual cost is at most
$3R_u$. Every residual request is eventually served exactly once.
\end{lemma}
\begin{proof}
We first prove capacity by integrating the load and then bound the
waiting time of each assigned request. For the final assignment to
root $u$, define
\[
 f_u(t)=\sum_{w\text{ assigned to }u}
           \frac{M_wR_u}{gy_w}\mathbf1_{I_w}(t).
\]
Homes containing a fixed point form a chain and are processed in ancestor
order. Every accepted cohort checks the existing chain sum plus its own
term. The root's initial term is $M_u/g\le1$. Once a new selected root
appears on the chain, subsequent descendants use that nearer root.
Consequently $f_u(t)\le1$ almost everywhere. Integration gives
\[
 \frac{R_u}{g}|\{r:r\text{ is assigned to }u\}|
   =\int_{I_u}f_u(t)\,dt\le R_u,
\]
which proves feasibility without splitting any request.

If $r$ has home $I_w\subseteq I_u$, then $y_w\le R_u$ and $b_w\ge b_u$.
Its actual arrival satisfies $a_r\ge b_w-y_w\ge b_u-R_u$. At time
$b_u+R_u$, its waiting is at most $2R_u$ and its radius is at most $R_u$.
The root batch therefore costs at most $3R_u$. Every residual cohort has
one owner with a finite prescribed service time, proving the last claim.
\end{proof}

\ifdefined\lineappendixmode
The complete event-driven rule is Algorithm~\ref{alg:path-load} in the main
text.
\else
Algorithm~\ref{alg:path-load} summarizes the complete event-driven rule.
\begin{algorithm}[!ht]
\caption{Path-load aggregation for LAD}\label{alg:path-load}
\begin{algorithmic}[1]
\Require Capacity $g$; one fixed scale $\lambda>0$; deterministic tie orders
\Statex Use $\lambda=1$ for the deterministic version; for Corollary~\ref{cor:line-random-shift},
draw $U$ uniformly from $[0,1)$ once and set $\lambda=2^U$.
\State Initially no requests, cohorts, or selected roots are stored.
\Statex \textbf{On arrival of request $r$ at $a_r$:}
\State $y_r\gets\min\{\lambda2^j:\lambda2^j\ge x_r\}$;
       $b_r\gets y_r(\lfloor a_r/y_r\rfloor+1)$
\State Store $r$ in the cohort with home $[b_r,b_r+y_r)$.
\Statex \textbf{At a virtual arrival time $b$:}
\For{nonempty homes $[b,b+y)$ in decreasing $y$}
  \State $Q\gets$ the requests stored for this home
  \While{$|Q|\ge g$}
    \State Serve the first $g$ requests of $Q$ at $b$ and remove them from $Q$.\label{line:full-cohort}
  \EndWhile
  \If{$Q\ne\varnothing$}
    \State Create cohort $v$ with $I_v=[b,b+y)$, $y_v=y$, and $M_v=|Q|$.
    \State $u\gets$ the selected strict ancestor of $I_v$ of minimum width, if any
    \If{$u$ exists}
      \State $L\gets\displaystyle\sum_{\substack{w\text{ assigned to }u\\I_v\subseteq I_w}}
        \frac{M_wR_u}{gy_w}+\frac{M_vR_u}{gy_v}$
    \EndIf
    \If{$u$ exists and $L\le1$}
      \State Assign the entire cohort $v$ to $u$.\label{line:admit-cohort}
    \Else
      \State Select $v$ as a root, set $R_v=y_v$, and assign $v$ to itself.\label{line:open-root}
      \State Schedule its service at $b+y_v$.
    \EndIf
  \EndIf
\EndFor
\Statex \textbf{At the service time $b_u+R_u$ of a root $u$:}
\State Serve all requests assigned to $u$ in one batch; retire $u$.
\end{algorithmic}
\end{algorithm}

\fi

\subsection{Two additive certificates}

A rejected cohort opens another batch, whose cost must be compared
with the offline solution. We use two request charges. The first
records the load already assigned to an owner; it will pay for
rejections caused by that load. The second pays for a new root when
its own cohort contributes a large load. Each charge is bounded inside
an arbitrary offline batch. Summing over the offline partition bounds the total charge, even when
the batches' arrival spans overlap.

We fix any feasible offline partition into batches. Each batch $B$
can be served at its latest arrival, giving minimum cost $X_B+D_B$,
where
\[
 X_B=\max_{r\in B}x_r,\qquad
 Y_B=\max_{r\in B}y_r,\qquad
 D_B=\max_{r\in B}a_r-\min_{r\in B}a_r.
\]
We also write $\alpha_B=|B\setminus P|/g\in[0,1]$ for the
fraction of capacity used by residual requests in $B$. The number of preprocessed requests in $B$ is at most $g(1-\alpha_B)$.
All residual homes represented in $B$ lie inside
$[\min_B a_r,\max_B a_r+2Y_B]$. Thus every disjoint collection of such
homes has total length at most $D_B+2Y_B$. This observation compares
virtual homes directly with the \emph{original} arrival span.

Each residual request assigned to $u$ receives weight
$\beta_r=R_u/g$; preprocessed requests receive weight zero.
We write $W=\sum_r\beta_r$ for their total weight.

\begin{lemma}\label{lem:line-beta}
For every feasible offline batch $B$,
\[
 \sum_{r\in B}\beta_r\le D_B+(2+\alpha_B)Y_B.
\]
\end{lemma}
\begin{proof}
Let $Q_1,\ldots,Q_m$ be the inclusion-maximal distinct residual homes
represented in $B$. They are disjoint, with total length at most
$D_B+2Y_B$. We split the charge according to whether a request's owner
is contained in its maximal home $Q_j$ or strictly contains $Q_j$.
The first group is bounded by the number of requests in $B$; the
second will be bounded by integrating load over the disjoint $Q_j$.

Owners of the first kind have width at most $Y_B$. Each of their
requests therefore has weight at most $Y_B/g$, and there are at most
$\alpha_Bg$ residual requests in $B$. Their total weight is at most
$\alpha_BY_B$.

All requests of the second kind within a given $Q_j$ have the same owner.
To see this, recall that a containing home is processed before its
descendants, including when their left endpoints coincide. Every strict
ancestor of $Q_j$ was therefore known when $Q_j$ was processed. Let $u$
be its nearest selected strict ancestor, if one exists. A descendant
home can be assigned above $Q_j$ only to $u$: any selected root between
that home and $u$ would instead become its nearest candidate owner.
Assignments are never revised, so roots selected later inside $Q_j$
do not change earlier owners. If $Q_j$ itself is selected, or if $u$
does not exist, the second group is empty. Otherwise,
the weight of this group is at most $\int_{Q_j}f_u(t)\,dt\le|Q_j|$.
Indeed, integrating the density of a whole assigned home contained in
$Q_j$ gives exactly its total $\beta$ weight; selecting only part of
that cohort reduces the weight. Summing over $j$ bounds the second
group by $D_B+2Y_B$. Adding the first group's $\alpha_BY_B$ gives
$D_B+(2+\alpha_B)Y_B$, as claimed.
\end{proof}

An owner may occur in several offline batches. The comparisons remain
additive because each batch sums only its own requests' $\beta$ weights;
the integral bounds these weights without assigning a new charge to
the owner each time.

For the analysis, we fix $\theta\in(0,1)$. A selected root $v$ is
\emph{heavy} if it has no selected ancestor or, writing $u$ for its
parent selected root, if
\[
 \frac{M_vR_u}{gy_v}\ge\theta.
\]
All other roots are \emph{light}. Each request in the cohort that
created a heavy root $v$ receives weight $\gamma_r=y_v/M_v$;
all other requests receive weight zero. Every residual
cohort is nonempty, and hence $M_v>0$. The total $\gamma$ weight equals the
total width of heavy roots.

\begin{lemma}\label{lem:line-gamma}
For every feasible offline batch $B$,
\[
 \sum_{r\in B}\gamma_r
   \le D_B+(2+\alpha_B/\theta)Y_B.
\]
\end{lemma}
\begin{proof}
We split the represented heavy roots into maximal and nonmaximal ones.
Disjoint homes bound the first group; the heavy-root inequality bounds
the charge per request in the second.
The inclusion-maximal heavy-root homes whose own-cohort requests occur
in $B$ are disjoint. Their total length is at most $D_B+2Y_B$, and their
represented $\gamma$ weight is at most this length, even when $B$ contains
only part of a cohort.

For a nonmaximal represented heavy root $v$, some represented heavy
ancestor $w$ exists. The parent selected root $u$ of $v$ lies in $I_w$,
and hence $R_u\le y_w\le Y_B$. Its heavy inequality gives
\[
 \gamma_r=\frac{y_v}{M_v}
       \le\frac{R_u}{g\theta}
       \le\frac{Y_B}{g\theta}.
\]
There are at most $\alpha_Bg$ such requests, so the nonmaximal roots
contribute at most $\alpha_BY_B/\theta$. Together with the maximal
roots' contribution $D_B+2Y_B$, this gives the claimed bound.
\end{proof}

\begin{lemma}\label{lem:line-light}
The sum of the widths of light roots is at most $W/(1-\theta)$.
\end{lemma}
\begin{proof}
We charge each light root to the load of its parent over its home.
These charges use disjoint intervals for a fixed parent, and different
parents use the weights of their own assigned requests.
Let $v$ denote a light immediate child root of $u$. Its own prospective term
is below $\theta$, but its admission failed. The already assigned
ancestor terms therefore sum to more than $1-\theta$. Every such term
contributes throughout $I_v$ and remains in the final density $f_u$.
Thus
\[
 |I_v|\le\frac1{1-\theta}\int_{I_v}f_u(t)\,dt.
\]
Immediate child-root intervals of $u$ are disjoint. Summing over light children and parents yields
\[
 \sum_{v\text{ light}}y_v
 \le\frac1{1-\theta}\sum_u\int_{I_u}f_u(t)\,dt
 =\frac1{1-\theta}\sum_r\beta_r.
\]
More explicitly, a request $r$ assigned to $u$ pays a light child
$v$ of $u$ the amount
\[
 \frac{\beta_r}{1-\theta}\,
       \frac{|H_r\cap I_v|}{|H_r|}.
\]
Its total contribution is at most $\beta_r/(1-\theta)$ by disjointness.
A grandchild of $u$ uses the requests assigned to its own parent.

\end{proof}

\subsection{Why Next Fit rather than First Fit?}
\label{subsec:line-first-fit}

The baseline in Section~\ref{subsec:class-baseline} loses a
logarithmic factor by processing distance classes independently.
Our algorithm combines classes, but the choice of owner also matters
for its analysis. Consider a \emph{First Fit}
variant that keeps the same preprocessing and load test, but tries all
selected containing roots in decreasing order of width. It joins the
first root that passes the test and opens a new root only if all fail.
This rule still preserves capacity. The example below shows that its
$\beta$ weights need not satisfy a constant-factor bound within an
arbitrary feasible offline batch. It is not a competitive lower bound
for First Fit.

\paragraph{A three-level example.}
Let $g=16$. At time $8-\varepsilon$, release eight requests at each
radius $8$, $4$, and $2$, where $0<\varepsilon<1/4$.
At virtual time $8$, they create roots with homes $[8,16)$, $[8,12)$,
and $[8,10)$. Each starts with load $1/2$. Every smaller seed cohort
is rejected by all larger roots, since its prospective load is at
least one before adding the existing load.

Inside $Q=[8,9)$, create one request for every home in the following
three levels. A request with home $[h,h+y)$ has radius $y$ and arrives
at $h-\varepsilon$.
\[
\begin{array}{c|c|c|c|c}
 \text{level}&\text{home width}&\text{number of homes}&
 \text{owner width}&\text{total }\beta\\ \hline
 0&1&1&8&1/2\\
 1&1/2&2&4&1/2\\
 2&1/4&4&2&1/2
\end{array}
\]
The first request adds load $8/(16\cdot1)=1/2$ to the largest root,
filling its load throughout $Q$. Each half-width home is therefore
rejected there, but adds load $4/(16\cdot(1/2))=1/2$ to the next root.
The quarter-width homes similarly join the smallest root. Ancestor
homes are processed first, so each rejection is already determined
when it occurs. Thus requests inside the same maximal home $Q$ use
three different owners strictly above $Q$, violating the unique-owner
property in Lemma~\ref{lem:line-beta}.

\paragraph{An unbounded loss in the certificate.}
For $m\ge2$, take $g=2^{m+1}$, $b=2^m$, and $\varepsilon=2^{-m}$.
Release $g/2$ seed requests at each radius $2^m,2^{m-1},\ldots,2$
at time $b-\varepsilon$. They create $m$ roots starting at $b$, each
with initial load $1/2$. Let $B$ consist of one additional request for
each dyadic home in $Q=[b,b+1)$ at levels $\ell=0,\ldots,m-1$:
\[
 H_{\ell,j}=[b+j2^{-\ell},b+(j+1)2^{-\ell}),
 \qquad j=0,\ldots,2^\ell-1.
\]
Its radius is $2^{-\ell}$ and its arrival is the left endpoint minus
$\varepsilon$. No cohort reaches capacity. As in the three-level
example, earlier levels fill each larger root on $Q$, while disjoint
homes at level $\ell$ fill the remaining load $1/2$ of the root of
width $2^{m-\ell}$. Each level contributes
$2^\ell\cdot2^{m-\ell}/g=1/2$ to the certificate. Thus
\[
 |B|=2^m-1<g,\qquad Y_B=1,\qquad
 D_B=1-2^{-(m-1)},\qquad \sum_{r\in B}\beta_r=m/2.
\]
Serving $B$ at its last arrival costs $Y_B+D_B<2$, so its $\beta$
weight exceeds its batch cost by $\Theta(\log g)$. This is a loss in
the batchwise certificate, not a competitive lower bound for First Fit:
the full input also contains the seed requests, whose offline cost
might support a different global accounting. Next Fit's unique-owner
property prevents this accumulation of area charges across scales.

\subsection{Combining preprocessing and residual costs}

Figure~\ref{fig:light-root-charge} illustrates the disjoint areas used
in Lemma~\ref{lem:line-light}. The accounting follows ownership: nested
selected roots do not repeatedly use one ancestor's requests.
\begin{figure}[tbp]
\centering
\begin{tikzpicture}[x=0.95cm,y=0.85cm,font=\small]
\node[anchor=east] at (1.65,5.3) {owner $u$};
\draw[figblue,very thick] (2,5.3)--(13,5.3);
\draw[figblue] (2,5.15)--(2,5.45) (13,5.15)--(13,5.45);
\node[above] at (7.5,5.3) {$I_u$};
\filldraw[fill=figgreen!12,draw=figgreen] (2,4.1) rectangle (4.75,4.7);
\filldraw[fill=figgreen!12,draw=figgreen] (7.5,4.1) rectangle (10.25,4.7);
\node at (3.375,4.4) {$I_{v_1}$};
\node at (8.875,4.4) {$I_{v_2}$};
\draw[paper arrow] (3.375,5.25)--(3.375,4.75);
\draw[paper arrow] (8.875,5.25)--(8.875,4.75);
\node[anchor=west,align=left,font=\footnotesize] at (10.55,4.4)
  {light children\\of $u$};
\filldraw[fill=figorange!10,draw=figorange,dashed] (2,3.05) rectangle (3.375,3.55);
\node at (2.6875,3.3) {$I_{v_3}$};
\draw[paper arrow,figorange] (2.6875,4.05)--(2.6875,3.6);
\node[anchor=west,font=\footnotesize] at (4,3.3)
  {If $v_3$ is light, its charge uses $f_{v_1}$, not $f_u$.};
\foreach \x in {2,4.75,7.5,10.25}{
  \draw[densely dotted] (\x,4.05)--(\x,3.65) (\x,2.95)--(\x,0.6);}
\fill[figgreen!18] (2,0.6) rectangle (4.75,2.25);
\fill[figgreen!18] (7.5,0.6) rectangle (10.25,2.05);
\draw[figblue,thick] (2,2.25)--(4.75,2.25)--(4.75,1.05)--(7.5,1.05)
  --(7.5,2.05)--(10.25,2.05)--(10.25,1.15)--(13,1.15);
\draw[dashed] (2,1.75)--(13,1.75);
\node[anchor=east,font=\footnotesize] at (1.8,1.75) {$1-\theta$};
\node[anchor=east] at (1.8,2.4) {$f_u(t)$};
\draw[paper arrow] (2,0.6)--(13.4,0.6) node[right] {$t$};
\node at (3.375,1.15) {$\int_{I_{v_1}} f_u$};
\node at (8.875,1.15) {$\int_{I_{v_2}} f_u$};
\node at (7.5,-0.15)
  {$\displaystyle (1-\theta)\sum_{v\text{ light child of }u}|I_v|
       \le\sum_v\int_{I_v}f_u(t)\,dt\le\int_{I_u}f_u(t)\,dt.$};
\end{tikzpicture}
\caption{Why the light-root charge does not repeat down the hierarchy
(schematic density). A light child's own term is below $\theta$, so
failed admission certifies $f_u>1-\theta$ throughout its interval.
Immediate child-root intervals are disjoint; the shaded areas therefore
use disjoint portions of the unique owner's density. A grandchild has
a different parent and uses that parent's assigned requests. The
parameter $\theta$ is used only in the analysis.}
\label{fig:light-root-charge}
\end{figure}

We combine three charges: preprocessing, heavy roots, and light roots.
By Lemmas~\ref{lem:line-load-capacity}--\ref{lem:line-light}, the last
two cost at most $3\sum_r\gamma_r$ and $3W/(1-\theta)$, respectively.
The factor three comes from the root-batch cost bound $3R_u$.
Inside an offline batch $B$, preprocessing contributes at most
$4(1-\alpha_B)Y_B$ by~\eqref{eq:line-prepaid}, while the two
certificates bound the residual charges. Keeping $\alpha_B$ in this
calculation accounts for the capacity shared by preprocessed and
residual requests. Substituting the certificate bounds, the total cost
is at most the sum over offline batches of
\[
 3\left(1+\frac1{1-\theta}\right)D_B+\left[
  3\left(2+\frac2{1-\theta}\right)+4
  +\alpha_B\left(\frac3\theta+\frac3{1-\theta}-4\right)
 \right]Y_B.
\]
The coefficient of $\alpha_B$ is positive. Maximizing over
$\alpha_B\in[0,1]$ yields the batchwise comparison
\begin{equation}\label{eq:line-direct-comparison}
 \ALG\le\sum_B\bigl(A(\theta)D_B+C(\theta)Y_B\bigr),
 \quad A(\theta)=3\left(1+\frac1{1-\theta}\right),
 \quad C(\theta)=3\left(2+\frac1\theta+\frac3{1-\theta}\right).
\end{equation}
The capacity constraint in this calculation belongs to the offline
batch $B$. Its preprocessed and residual requests occupy disjoint
portions of that capacity, although the online algorithm may serve
them in different batches.

The choice $\theta=1/(1+\sqrt3)$ minimizes $C(\theta)$ and gives
\[
 A(\theta)=6+\sqrt3,
 \qquad C(\theta)=C_*=18+6\sqrt3.
\]
For already scaled-dyadic radii, $Y_B=X_B$, and hence
\eqref{eq:line-direct-comparison} proves the $C_*$ bound. For arbitrary
radii, $Y_B<2X_B$, proving the $2C_*$ bound of
Theorem~\ref{thm:line-path-load}. The threshold $\theta$ changes only the
proof, not any online decision. All comparisons hold for each realized
input and each feasible offline partition, which also establishes the
claimed adaptive-adversary scope.

\subsection{Randomizing the scale}

\ifdefined\lineappendixmode
\begin{proof}[Proof of Corollary~\ref{cor:line-random-shift}]
\else
\begin{corollary}\label{cor:line-random-shift}
Against an oblivious adversary, LAD with
capacity $g$ admits a randomized online
algorithm with expected competitive ratio
\[
 \frac{18+6\sqrt3}{\ln2}<40.962.
\]
\end{corollary}
\begin{proof}
\fi
The algorithm draws $U$ uniformly from $[0,1)$ once, sets
$\lambda=2^U$, and uses this scaled dyadic grid for both radii and time.
For any fixed $x>0$, its rounded radius is
\[
 y_U(x)=2^{U+\lceil\log_2x-U\rceil},
 \qquad \E[y_U(x)]=x\int_0^1 2^z\,dz=\frac{x}{\ln2}.
\]
We fix an input independent of $U$ and an optimal partition of the
original requests. The common monotone rounding map gives
$Y_B=y_U(X_B)$. Applying~\eqref{eq:line-direct-comparison} to this
same partition for every $U$ and taking expectations gives
\[
 \E[\ALG]\le\sum_B\left((6+\sqrt3)D_B+
                    \frac{C_*}{\ln2}X_B\right)
           \le\frac{C_*}{\ln2}\OPT,
\]
which proves the corollary by linearity of expectation.
\end{proof}

The expectation uses an input fixed independently of the scale.
For adaptive inputs, the deterministic $2C_*$ guarantee still holds
for every realized scale. The randomized implementation uses an
ideal-real computation model: it samples the common scale exactly and
performs arithmetic, comparisons, and rounding on real values exactly.
It uses polynomially many such operations. We do not claim a finite-bit
implementation of this exact continuous random scale.

\ifdefined\lineappendixmode
\else
\subsection{Arbitrary request sizes}
\label{subsec:arbitrary-sizes}

We now allow each request to have an indivisible size $s_r$ revealed at
arrival. A batch of
capacity $K>0$ must satisfy $\sum_{r\in Q}s_r\le K$,
where $s_r\in[0,K]$. Requests cannot be split, and size affects feasibility
only: a batch still pays its service cost plus its maximum waiting time.
Zero-size requests may therefore have positive cost. Unit sizes with
$K=g$ recover the main model.
With $x_r=1$, $K=1$, and positive sizes, this is bin packing with per-bin
maximum delay~\cite{lu2026binmaximum}; that special case admits smaller
competitive constants than the guarantees below for arbitrary radii.

\begin{corollary}\label{cor:size-path-load}
The deterministic and randomized guarantees of
Theorem~\ref{thm:line-path-load} and Corollary~\ref{cor:line-random-shift}
hold for arbitrary indivisible sizes, including zero sizes. The
deterministic algorithm is polynomial-time for rational input. The
randomized algorithm uses polynomially many operations in the ideal-real
model specified in Section~\ref{sec:line-path-load}. Neither needs a minimum-size assumption.
In particular, the deterministic competitive ratio is $36+12\sqrt3$.
\end{corollary}

Only preprocessing and the units of path load change. Within each
newly revealed cohort, serve requests larger than $K/2$ as singletons.
Process the others by Next Fit in a fixed order, closing and serving
the current bin when the next request does not fit. Retain the final
nonempty bin as the residual cohort, even when its load is exactly $K$.
Every served bin has load greater than $K/2$, so its cost is paid by
the size-weighted volume bound, with a constant loss.

Let $M_v$ now denote the total size of residual cohort $v$, rather than
its cardinality. Its contribution to the path load of an owner $u$ is
$M_vR_u/(Ky_v)$. Integrating this load over
the owner's interval bounds the total assigned size by $K$. Thus the
same admission rule preserves capacity without splitting a cohort or
any request. The light-root argument also uses the same path load;
the certificates $\beta$ and $\gamma$ distribute their weights by size.
Appendix~\ref{sec:arbitrary-sizes} gives the weighted proof, including
zero-size requests, and the pseudocode. The algorithm requires no
optimum bin packing.
\fi

\let\lineappendixmode\undefined
\def\adviceappendixmode{1}
\ifdefined\adviceappendixmode
\section{The full online-advice monitor}
\label{sec:advice-full}
This appendix gives the complete monitor, its residual-advice proof, and the
algebra behind the two guarantees stated in
Theorem~\ref{thm:feasible-stream-monitor}.
\else
\section{Learning-augmented capacitated LAD}
\label{sec:online-advice}
Section~\ref{sec:line-path-load} gives a constant-competitive algorithm
without advice. We now seek to follow a good advised schedule more
closely while retaining a worst-case guarantee. We use the online
revelation of advice and the consistency benchmark of Lu et
al.~\cite{lu2026learninglad} for additive-delay LAD. With capacity,
each suggestion must also
specify which requests it serves: at time $t$, an advice event reveals a
service $(Q,t)$ and the identities of its already arrived requests. The
algorithm sees no future advice events or total advice cost. The provider
may know the full input; only when its advice is revealed is restricted.

As in~\cite{lu2026learninglad}, consistency $1+\eps$ means
$\ALG(I)\le(1+\eps)\cost(\ADV)$ for \emph{every} complete,
capacity-feasible advice stream $\ADV$ on $I$. Robustness $R$ requires
$\ALG(I)\le R\OPT(I)$ even for inaccurate advice.
Invalid references can be ignored and no physical action may violate
capacity. Consistency compares with the supplied feasible schedule, even
when it is far from optimal.

The new difficulty is accounting for requests left pending when switching
between advice and a robust baseline. Under additive delay, waiting costs
are charged separately to requests. Under maximum delay, one advice batch
may share a waiting cost that the algorithm pays across several batches.
Capacity can also force a switch to clear old requests in several
services. Charging each clearance to the original advice batch would
count its cost repeatedly. We therefore need switching rules that control
the pending requests and charges that remain additive across switches.
The monitor below gives the following guarantee.

\begin{samepage}
\begin{theorem}\label{thm:feasible-stream-monitor}
For LAD with capacity $g$, let $B$ denote
a deterministic $R$-competitive online algorithm. For every
$0<\eps\le1$, there is a deterministic online algorithm satisfying
\begin{align}
 \ALG(I)&\le(1+\eps)\cost(\ADV),
       \label{eq:stream-consistency}\\
 \ALG(I)&\le
 \left(\frac{32}{\eps^2}+\frac{72}{\eps}+36\right)R\OPT(I).
       \label{eq:stream-robustness}
\end{align}
The first bound holds for every complete, capacity-feasible advice stream
$\ADV$ revealed at its service times. The second holds independently of advice
quality, including when advice is incomplete. The algorithm needs no
future advice partition or total advice cost. For rational input and
$\eps$, the overhead is polynomial in the processed input and advice
description if $B$ is polynomially simulable.
\end{theorem}
\end{samepage}

With the LAD algorithm from Section~\ref{sec:line-path-load} as the
baseline, robustness is $O(\eps^{-2})$.

\paragraph{The switching idea.}
The algorithm simulates the advice and the baseline throughout the
input. It measures each schedule by the cost already paid plus lower
bounds on the service and waiting costs of its pending requests.
Geometrically increasing budgets determine which schedule to follow.
Advice is preferred whenever its budget permits; a smaller budget for
the baseline limits spending away from advice. At a switch, clearing
all old pending requests ensures that the new schedule can serve every
request still waiting for the algorithm.

The main issue is paying for these clearances while keeping consistency
close to one. When leaving advice, the waiting charge is paid by the
decrease in its remaining cost after deleting the requests just cleared.
Different switches use successive decreases, rather than repeatedly
charging an entire advice batch. The baseline budget pays the additional
service cost. When returning to advice, the algorithm also clears
requests that advice served while the baseline was active; these are
charged within the baseline block. This restores the residual-advice comparison at every return.

\fi

\subsection{Packing costs and virtual schedules}
The algorithm serves radius-zero requests immediately for free and
removes them from advice. We therefore assume positive radii below.
An input containing only radius-zero requests costs nothing.

\paragraph{Separate packing costs.}
For a pending set $P$ at time $t$, we write $w_r(t)=t-a_r$ and define
\begin{equation}\label{eq:stream-separate-packings}
 S(P)=\min_{\Pi}\sum_{G\in\Pi}\max_{r\in G}x_r,
 \qquad
 D(P,t)=\min_{\Pi}\sum_{G\in\Pi}\max_{r\in G}w_r(t),
\end{equation}
where each minimum is over partitions into groups of at most $g$
requests. The two minimizing partitions may differ: $S$ ignores waiting,
and $D$ ignores distance. Nevertheless, any feasible partition pays at
least $S(P)$ in service cost and $D(P,t)$ in current waiting cost, so
their sum is a lower bound on clearing $P$. Both values are zero on the
empty set. For nonnegative weights
$v_r$, sorting them in nonincreasing order and grouping consecutive
$g$ requests minimizes the sum of group maxima. Indeed, for every
$z\ge0$ any partition has at least
$\lceil|\{r:v_r>z\}|/g\rceil$ groups of maximum weight greater than
$z$, and the sorted partition attains all these bounds. Integration over
$z$ proves optimality. Thus $S$ and $D$ are exactly computable by sorting;
they are also monotone under taking subsets.

\begin{samepage}
\begin{lemma}\label{lem:stream-two-clearances}
There are polynomial-time procedures that clear $P$ at time $t$ with
cost at most
\begin{equation}\label{eq:stream-ordinary-clearance}
 2\bigl(S(P)+D(P,t)\bigr),
\end{equation}
and, for any $\eta>0$, with cost at most
\begin{equation}\label{eq:stream-biased-clearance}
 (1+\eta)D(P,t)+(1+1/\eta)S(P).
\end{equation}
Every output batch has at most $g$ requests.
\end{lemma}
\end{samepage}

The ordinary procedure sorts requests by nonincreasing $\max\{x_r,w_r(t)\}$ and
groups consecutive $g$ requests. The biased procedure separates
$L=\{r:x_r\le\eta w_r(t)\}$ from its complement, sorting $L$ by
nonincreasing age and its complement by nonincreasing radius before
packing each separately.
We use \textsc{PackValue} for the optimum sorted packing value and
\textsc{Clear} for either clearance procedure. Their pseudocode and the
proof of Lemma~\ref{lem:stream-two-clearances} appear in
Appendix~\ref{sec:advice-analysis}.

\paragraph{Two virtual schedules.}
The algorithm continuously simulates two schedules on the original
input: the baseline $B$ and the advice schedule $A$. A revealed advice
proposal is executed in the simulation of $A$. Neither simulation is
restarted at a switch. We call their states \emph{virtual}, and use
\emph{physical} for the services actually executed by our algorithm.
Let $p_X(t)$ and $P_X(t)$ denote the paid cost and pending set of
$X\in\{A,B\}$, respectively. Its score is
\begin{equation}\label{eq:stream-source-score}
 H_X(t)=p_X(t)+S(P_X(t))+D(P_X(t),t).
\end{equation}
The score includes both the cost already paid by the virtual schedule
and lower bounds on the cost of its pending requests. All three terms
are available online. For malformed advice, the algorithm discards
identifiers outside $P_A$ and ignores empty or capacity-infeasible
proposals. This leaves feasible advice unchanged. Physical services leave both virtual pending sets unchanged.

\begin{lemma}\label{lem:stream-source-score}
The score $H_X$ is nondecreasing. If $X$ is complete and feasible, then
$H_X(t)\le C_X:=\cost(X)$ for all $t$.
Moreover, for any set $P$ of already arrived requests,
\begin{equation}\label{eq:stream-spatial-certificate}
 S(P)\le p_B(t)+S(P_B(t))\le H_B(t).
\end{equation}
\end{lemma}

\subsection{The switching rule}
We set
\begin{equation}\label{eq:stream-parameters}
 \eta=\frac{\eps}{2},\qquad
 \delta=\frac{\eps}{4(3+2/\eps)}.
\end{equation}
At the first positive-radius arrival block $J$, the algorithm
initializes $b_0=2S(J)$ and uses the levels $b_n=2^n b_0$.
The level never decreases. At level $b$, advice $A$ is eligible when
$H_A<b$, and the baseline $B$ is eligible when $H_B<\delta b$.
The algorithm prefers $A$ whenever it is eligible. If neither is
eligible, it doubles $b$ until one becomes eligible. Initially $A$
is eligible because all requests in $J$ have age zero. The smaller
baseline budget limits how much the algorithm spends away from advice.

The chosen schedule is called the \emph{active source}. The algorithm
follows each of its services on the requests that are still physically
pending, omitting empty restrictions. The other source continues only
in simulation. When the active score reaches its cap, the algorithm
selects again. A higher level with the same active source requires no
clearance. An actual source change first clears all old physical pending
requests, using~\eqref{eq:stream-biased-clearance} for $A\to B$ and
\eqref{eq:stream-ordinary-clearance} for $B\to A$
(lines~\ref{line:biased-switch} and~\ref{line:ordinary-switch} of
Algorithm~\ref{alg:monitor-select}). The latter clearance includes
requests that virtual advice served while $B$ was active.

\begin{algorithm}[!ht]
\caption{Advice monitor: state, scores, and budget selection}\label{alg:monitor-select}
\begin{algorithmic}[1]
\Require $0<\eps\le1$; continuously simulated advice $A$ and baseline $B$
\State $\eta\gets\eps/2$; $\delta\gets\eps/[4(3+2/\eps)]$
\State $S(P)\gets\Call{PackValue}{P,x,g}$;
       $D(P,t)\gets\Call{PackValue}{P,t-a,g}$
\State For $Y\in\{A,B\}$ maintain its own pending set $P_Y$ and paid cost $p_Y$.
\State $H_Y(t)\gets p_Y+S(P_Y)+D(P_Y,t)$; maintain physical pending $P$ separately.
\State Serve radius-zero arrivals immediately and remove their identifiers from advice.
\State At the first positive-radius arrival block $J$, admit $J$ to all three queues;
       set $b\gets2S(J)$ and active source $X\gets A$ before service events.
\Function{Select}{$h_A,h_B,t$}
 \State $X_{\rm old}\gets X$
 \While{$h_A\ge b$ and $h_B\ge\delta b$}
  \State $b\gets2b$
 \EndWhile
 \If{$h_A<b$}
  \State $X\gets A$
 \Else
  \State $X\gets B$
 \EndIf
 \If{$X\ne X_{\rm old}$}
  \If{$X_{\rm old}=A$}
   \State \Call{Clear}{$P,t,\textit{biased},\eta$}\label{line:biased-switch}
  \Else
   \State \Call{Clear}{$P,t,\textit{ordinary},1$}\label{line:ordinary-switch}
  \EndIf
 \EndIf
 \Comment{No cleanup on a same-source level change}
\EndFunction
\State Define $\operatorname{cap}(A)=b$ and $\operatorname{cap}(B)=\delta b$.
\end{algorithmic}
\end{algorithm}

Algorithms~\ref{alg:monitor-select} and~\ref{alg:monitor-events}
specify budget selection and event processing. At a common timestamp,
arrivals precede virtual services in a fixed order, followed by any
remaining timer. A discrete event that exhausts the active cap triggers
selection using tentative post-event scores, before the event is
committed. This prevents a jump from being charged to an exhausted
budget. On a source change, clearing the old physical queue restores
the pending-set inclusion proved below: new arrivals enter after
clearance, and a triggering virtual service is committed without being
executed physically. Between events, the active score grows at rate
$\lceil|P_X|/g\rceil$, which determines its next cap timer.
Eligibility is strict, and hence equality exhausts a cap. Throughout a block
the active score is at most its cap; the inequality may be an equality
just before a simultaneous event. A \emph{block} is a maximal interval
with the same active source, and may span several levels.

\begin{algorithm}[!ht]
\caption{Advice monitor: arrivals, virtual services, and timers}\label{alg:monitor-events}
\begin{algorithmic}[1]
\Require State from Algorithm~\ref{alg:monitor-select}; arrival events precede services and timers at ties
\Statex \textbf{Arrival block $J$ at time $t$, after initialization:}
\State $\widehat H_Y\gets p_Y+S(P_Y\cup J)+D(P_Y\cup J,t)$ for $Y\in\{A,B\}$
\If{$\widehat H_X\ge\operatorname{cap}(X)$}
 \State \Call{Select}{$\widehat H_A,\widehat H_B,t$}
       \Comment{Cleanup sees $P$, not $P\cup J$}
\EndIf
\State Add $J$ to $P,P_A,P_B$; then process the time-$t$ virtual services.\label{line:admit-arrivals}
\Statex \textbf{Virtual service $(Y,G,t)$, in fixed order at ties:}
\State If $Y=A$, keep only distinct identifiers of $G$ in $P_A$;
       discard empty or infeasible proposals.
\State $\widehat H_Y\gets p_Y+c(G,t)+S(P_Y\setminus G)+D(P_Y\setminus G,t)$
\State $\widehat H_Z\gets H_Z(t)$ for the other source $Z$; $X_{\rm old}\gets X$
\If{$Y=X$ and $\widehat H_Y\ge\operatorname{cap}(X)$}
 \State \Call{Select}{$\widehat H_A,\widehat H_B,t$}
       \Comment{Before committing this virtual service}
\EndIf
\If{$Y=X_{\rm old}$ and $X=X_{\rm old}$}
 \State Serve $G\cap P$ physically if nonempty; remove it from $P$.
\EndIf
\State $p_Y\gets p_Y+c(G,t)$; $P_Y\gets P_Y\setminus G$\label{line:commit-service}
       \Comment{Always commit the virtual event}
\Statex \textbf{Waiting timer, after arrivals and virtual services at time $t$:}
\If{$H_X(t)\ge\operatorname{cap}(X)$}
 \State \Call{Select}{$H_A(t),H_B(t),t$}
\EndIf
\Statex \textbf{After each event or selection:}
\State If $P_X\ne\varnothing$, set the next cap timer to
 $t+[\operatorname{cap}(X)-H_X(t)]/\lceil|P_X|/g\rceil$;
 otherwise cancel it.
\State Continue both virtual schedules without restarting either one.
\end{algorithmic}
\end{algorithm}

Figure~\ref{fig:advice-switching} summarizes the two kinds of cleanup
and the order of operations at a triggering event.
\begin{figure}[tbp]
\centering
\begin{tikzpicture}[x=1cm,y=0.9cm,font=\small]
\node[paper box,draw=figblue,fill=figblue!8,text width=4.4cm,minimum height=1.35cm]
  (A) at (3,3) {Follow advice $A$\\$P\subseteq P_A$\\eligible when $H_A<b$};
\node[paper box,draw=figorange,fill=figorange!8,text width=4.4cm,minimum height=1.35cm]
  (B) at (10.5,3) {Follow baseline $B$\\$P\subseteq P_B$\\eligible when $H_B<\delta b$};
\draw[paper arrow,figblue] (A.north)--(3,4.7)--(10.5,4.7)--(B.north);
\node[above,align=center] at (6.75,4.75)
  {Clear old pending with biased packing\\
   $(1+\eta)D+(1+1/\eta)S$};
\draw[paper arrow,figorange] (B.south)--(10.5,1.3)--(3,1.3)--(A.south);
\node[below,align=center] at (6.75,1.2)
  {Clear all old pending for at most $2(S+D)$\\
   Include requests already served by virtual advice.};
\node[align=center,font=\footnotesize] at (6.75,-0.5)
  {Both virtual schedules keep running. Prefer $A$ when eligible.\\
   If neither is eligible, double $b$; keeping the same source causes no cleanup.};
\draw[black!25] (0.5,-1.25)--(13,-1.25);
\node[anchor=west,font=\small\bfseries] at (0.5,-1.75) {At a triggering arrival or active service:};
\node[paper box,text width=2.45cm,minimum height=0.95cm,font=\footnotesize] (e1) at (1.65,-2.85)
  {Tentative\\post-event scores};
\node[paper box,text width=2.2cm,minimum height=0.95cm,font=\footnotesize] (e2) at (4.85,-2.85)
  {Select level\\and source};
\node[paper box,text width=2.8cm,minimum height=0.95cm,font=\footnotesize] (e3) at (8.25,-2.85)
  {Clear old pending\\if switching};
\node[paper box,text width=2.15cm,minimum height=0.95cm,font=\footnotesize] (e4) at (11.65,-2.85)
  {Commit\\revealed event};
\draw[paper arrow] (e1)--(e2);
\draw[paper arrow] (e2)--(e3);
\draw[paper arrow] (e3)--(e4);
\end{tikzpicture}
\caption{The advice combiner separates virtual histories from physical
service. Here $P$ is the physical pending set, and the displayed packing
costs are evaluated on the old set cleared at a switch. Returning to
$A$ removes stale identifiers before following advice again. At a jump,
selection uses tentative post-event scores, while cleanup uses
pre-event pending. A triggering service is committed virtually but is
not replayed physically after a source change.}
\label{fig:advice-switching}
\end{figure}

\begin{lemma}\label{lem:stream-pending-invariant}
During an $X$ block, the physical pending set is a subset of $P_X$.
Each source change clears all old physical pending before the new block
starts. If an arrival causes the change, the new arrivals belong to both
virtual pending sets when they enter the new block.
\end{lemma}
\begin{proof}
The statement holds initially. An arrival adds the same identifiers to
both virtual queues and the physical queue. An active virtual service
removes all of its physically pending identifiers. Inactive services do
not affect the active queue. At a source change, clearing old physical
pending restores the inclusion for either new source; the event rules
ensure that a triggering service is not missed by a nonempty new
queue: arrivals enter after clearance, and virtual services are committed
after selection (lines~\ref{line:admit-arrivals} and
\ref{line:commit-service} of Algorithm~\ref{alg:monitor-events}).
This proves the invariant.
\end{proof}

\subsection{Residual advice and the waiting charge}
We now fix any complete feasible advice schedule, with original
batches $(G_j,T_j)$. For consistency, we need to pay for followed
advice services and for the waiting part of each $A\to B$ clearance.
Both payments come from deleting physically served requests from the
original advice batches. For the set $Z$ of physically served request
identifiers, the remaining advice cost is
\begin{equation}\label{eq:stream-residual-potential}
 \Phi(Z)=\sum_j c(G_j\setminus Z,T_j),\qquad c(\varnothing,T_j)=0.
\end{equation}
The sum includes original services whose times are already in the past.
It is a potential used only in the proof, not a future schedule available
to the algorithm. Deleting identifiers cannot increase any summand.
Initially $\Phi(\varnothing)=C_A$.

\begin{lemma}\label{lem:stream-residual-credit}
Let $F_A$ denote the total cost of physical services following source $A$.
At each $A\to B$ change $e$, let $P_e$ denote the old physical pending set
cleared at time $t_e$. Then
\begin{equation}\label{eq:stream-residual-credit}
 F_A+\sum_e D(P_e,t_e)\le C_A.
\end{equation}
\end{lemma}
\begin{proof}
We show that each followed advice service reduces $\Phi$ by its
physical cost, and each $A\to B$ clearance reduces $\Phi$ by at least
the corresponding $D$ value. Since $\Phi$ starts at $C_A$ and never
increases, these decreases can be summed.

At a followed advice service $(G_j,T_j)$, its remaining identifiers
$G_j\setminus Z$ are exactly its physical restriction. Deleting them
reduces the $j$th summand of~\eqref{eq:stream-residual-potential} by its
actual service cost.

For an $A\to B$ change at time $t$, we write $P=P_e$. By
Lemma~\ref{lem:stream-pending-invariant}, every request in $P$ is still
pending in virtual advice. Thus each affected original batch has
$T_j\ge t$. We write $Q_j=G_j\cap P$. After deleting $P$, any surviving
identifier in $G_j\setminus Z$ has release at least $t$: all earlier
arrivals were either served physically before this change or belong to
$P$. New arrivals that trigger the change have release exactly $t$.

If an affected batch remains nonempty, its oldest surviving release
therefore advances to at least $t$, reducing its delay cost by at least
$t-\min_{r\in Q_j}a_r$. Its service cost does not increase. If the batch
becomes empty, its entire remaining cost is at least the same amount,
since $T_j\ge t$. The nonempty $Q_j$ form a capacity-feasible partition
of $P$, and hence the total potential drop is at least $D(P,t)$.

All other physical actions, including $B\to A$ cleanup, cause additional
nonnegative potential drops. Summing the charged drops proves the lemma.
In particular, repeated charges involving one original advice batch use
successive decreases of its residual cost, never its original cost again.
\end{proof}

During a $B$ block, advice may serve a request that remains physically
pending. Deleting such an identifier later may affect a past advice
batch. We assign no waiting credit to that deletion: the entire physical
service is paid by the $B$-block bound below. Before returning to $A$, the
full cleanup removes all such identifiers and restores the inclusion in
Lemma~\ref{lem:stream-pending-invariant}.

\suppressfloats[t]
\begin{figure}[!ht]
\centering
\begin{tikzpicture}[x=1cm,y=0.85cm,font=\small]
\node[anchor=west,font=\small\bfseries] at (0,5.2)
  {(a) One advice batch surviving an $A\to B$ cleanup};
\draw[paper arrow] (2,4.4)--(13.4,4.4) node[right] {time};
\foreach \x/\lab in {2/a_{\rm old},6/t,9/a_{\rm new},12.5/T_j}{
  \draw (\x,4.28)--(\x,4.52);
  \node[above] at (\x,4.55) {$\lab$};}
\node[anchor=east] at (1.8,3.55) {before};
\filldraw[draw=figblue,fill=figblue!10] (2,3.3) rectangle (12.5,3.8);
\node at (7.25,3.55) {delay $T_j-a_{\rm old}$};
\node[anchor=east] at (1.8,2.55) {after};
\fill[figgreen!17] (2,2.3) rectangle (6,2.8);
\fill[black!8] (6,2.3) rectangle (9,2.8);
\draw[dashed,figgreen] (2,2.3) rectangle (9,2.8);
\filldraw[draw=figblue,fill=figblue!10] (9,2.3) rectangle (12.5,2.8);
\node[font=\footnotesize] at (4,2.55) {charged: $t-a_{\rm old}$};
\node[font=\footnotesize] at (7.5,2.55) {extra drop};
\node[font=\footnotesize] at (10.75,2.55) {$T_j-a_{\rm new}$};
\node[align=center] at (7.25,1.65)
  {$a_{\rm new}\ge t$ implies delay decrease $\ge t-a_{\rm old}$.};
\node[anchor=west,font=\small\bfseries] at (0,0.75)
  {(b) Successive deletions use successive potential drops};
\node[paper box,fill=figblue!8] (p0) at (2,-0.1) {$\Phi(Z_0)=C_A$};
\node[paper box,fill=figblue!8] (p1) at (5.5,-0.1) {$\Phi(Z_1)$};
\node[paper box,fill=figblue!8] (p2) at (9,-0.1) {$\Phi(Z_2)$};
\node (p3) at (12,-0.1) {$\cdots\ \ge0$};
\draw[paper arrow] (p0)--node[above] {$-\Delta_1$}(p1);
\draw[paper arrow] (p1)--node[above] {$-\Delta_2$}(p2);
\draw[paper arrow] (p2)--node[above] {$-\Delta_3$}(p3);
\node at (7,-1.1)
  {$\displaystyle F_A+\sum_{e:A\to B}D(P_e,t_e)
       \le\sum_\ell\Delta_\ell\le C_A.$};
\end{tikzpicture}
\caption{Residual charging pays for elapsed waiting only once.
In (a), the original advice batch is served at $T_j\ge t$; after all
old physical pending requests are deleted, its oldest surviving arrival
is at least $t$. The diagram shows the nonempty-residual case, and
service cost can only decrease as well. In (b), the potential retains
all original advice batches, including past ones. Deletions in $B$ blocks yield nonnegative drops but no waiting credit;
the block budgets pay their physical cost.}
\label{fig:residual-advice-charge}
\end{figure}

Figure~\ref{fig:residual-advice-charge} shows the waiting credit at one
$A$ exit and how credits telescope.

\subsection{Consistency and robustness}
We next pay for the service part of each $A$ exit and for complete
$B$ blocks. The following budget bounds are proved in
Appendix~\ref{sec:advice-analysis}. Their key property is that a budget
level cannot be reused by distinct blocks of the same source.

Let $\mathcal E$ denote the set of $A\to B$ changes, and let
$b_e^+$ denote the newly selected $B$ level at change $e$. The service-cost certificate gives
\begin{equation}\label{eq:stream-exit-space}
 S(P_e)\le H_B^{\rm post}<\delta b_e^+.
\end{equation}
This uses the post-event baseline state to cover all old requests, even
when the triggering event is an arrival or a service. The cleanup itself
uses the pre-event physical pending set.

For a block $J$, let $b_J$ denote its last occupied level and $F_J$
the cost of its followed physical services, excluding exit cleanup.
If the block ends at a switch, $P_J$ and $t_J$ denote the old pending
set and exit time. For a final block without exit cleanup, we use
$D(P_J,t_J)=0$ below.
\begin{samepage}
\begin{lemma}\label{lem:stream-block-budgets}
For an $A$ block,
\begin{equation}\label{eq:stream-A-block}
 F_J+(1+\eta)D(P_J,t_J)\le(1+\eta)b_J.
\end{equation}
For a $B$ block, the followed services together with its exit cleanup
cost at most $2\delta b_J$.
\end{lemma}
\end{samepage}

\begin{lemma}\label{lem:stream-level-accounting}
At a fixed level there is at most one $A$ block and one $B$ block.
The last levels of distinct blocks of the same source are distinct,
as are the newly selected levels $b_e^+$ at $A\to B$ changes. If advice
is complete and feasible, then
\begin{equation}\label{eq:stream-advice-level-sums}
 \sum_{e\in\mathcal E}b_e^+\le2C_A,
 \qquad \sum_{J:\,B\text{ block}}b_J\le2C_A.
\end{equation}
For arbitrary advice, the largest level obeys
\begin{equation}\label{eq:stream-largest-level}
 b_{\max}\le2C_B/\delta,
\end{equation}
and the sum of all opened levels is at most $2b_{\max}$.
\end{lemma}

\begin{proof}[Proof of Theorem~\ref{thm:feasible-stream-monitor}]
All physical batches are feasible by construction. For consistency,
we partition physical cost into followed $A$ services, $A\to B$
cleanups, and complete $B$ blocks including their exit cleanups.
Lemma~\ref{lem:stream-residual-credit} pays for followed advice services
and the waiting part of leaving $A$. The service part is bounded
by~\eqref{eq:stream-exit-space}, and Lemma~\ref{lem:stream-block-budgets}
pays for complete $B$ blocks. Using the two geometric sums in
Lemma~\ref{lem:stream-level-accounting} gives
\begin{align*}
 \ALG
 &\le F_A+(1+\eta)\sum_{e\in\mathcal E}D(P_e,t_e)
       +(1+1/\eta)\delta\sum_{e\in\mathcal E}b_e^+
       +2\delta\sum_{J:\,B\text{ block}}b_J\\
 &\le\bigl[1+\eta+2\delta(3+1/\eta)\bigr]C_A
  =(1+\eps)C_A.
\end{align*}
The residual-advice bound controls the first two terms by
$(1+\eta)C_A$. The $B$-block term includes the cost of requests already
served by virtual advice; these require no further advice charge.

For robustness, we sum~\eqref{eq:stream-A-block} over $A$ blocks,
bound their additional service-cost cleanup terms by
\eqref{eq:stream-exit-space}, and apply the $B$-block bound. These
comparisons use no bound on $C_A$. The last levels of $A$ blocks, the
last levels of $B$ blocks, and the newly selected levels at $A\to B$
switches each form a sequence without repetition. Bounding each sum
by $2b_{\max}$ gives
\[
 \ALG
 \le2\bigl[(1+\eta)+\delta(3+1/\eta)\bigr]b_{\max}
 \le4\left[\frac{1+\eta}{\delta}+3+\frac1\eta\right]C_B
  =\left(\frac{32}{\eps^2}+\frac{72}{\eps}+36\right)C_B.
\]
Since $C_B\le R\OPT(I)$, this is~\eqref{eq:stream-robustness}.
Completion and the polynomial-time implementation are verified in
Appendix~\ref{sec:advice-analysis}.
\end{proof}

\subsection{A lower bound on robustness}
The dependence on $\eps$ is necessary even on inputs that fit in one
batch at one location. Since these are also LAD inputs, the following
lower bound applies to the general advice problem.

\begin{theorem}\label{thm:stream-lower}
Fix $g\ge1$ and $\eps\ge0$. For maximum delay at one location, suppose a
deterministic algorithm is $(1+\eps)$-consistent against every feasible
online-revealed advice stream and $R$-robust. Already on inputs with at most
$g$ requests it must satisfy
\begin{equation}\label{eq:stream-lower}
 R\ge 1+\frac{g-1}{1+\eps(g-1)}.
\end{equation}
Consequently, a robustness bound uniform in $g$ must satisfy
$R\ge1+1/\eps$ for $\eps>0$; at $\eps=0$ no finite such bound is possible.
\end{theorem}

The proof is in Appendix~\ref{sec:advice-analysis}; the optimal dependence
on $\eps$ is open.

\let\adviceappendixmode\undefined
\section{Supporting proofs for the online-advice algorithm}
\label{sec:advice-analysis}
\label{sec:feasible-advice-monitor}
Section~\ref{sec:online-advice} gives the algorithm, the residual-advice
argument, and the competitive bounds. This appendix proves the packing
and budget lemmas used there, verifies completion and implementation,
and proves the robustness lower bound.

\subsection{Packing procedures and scores}
\begin{proof}[Proof of Lemma~\ref{lem:stream-two-clearances}]
For~\eqref{eq:stream-ordinary-clearance}, the procedure sorts requests
by $q_r=\max\{x_r,w_r(t)\}$ and groups consecutive $g$ requests.
For every $z$, the number of groups whose maximum $q$ exceeds $z$ is
\[
 \left\lceil\frac{|\{r:q_r>z\}|}{g}\right\rceil
 \le
 \left\lceil\frac{|\{r:x_r>z\}|}{g}\right\rceil+
 \left\lceil\frac{|\{r:w_r(t)>z\}|}{g}\right\rceil.
\]
Integrating bounds the sum of group maxima of $q$ by $S(P)+D(P,t)$.
Each actual batch costs at most twice its maximum $q$.

For~\eqref{eq:stream-biased-clearance}, let $L$ denote the set
$\{r\in P:x_r\le\eta w_r(t)\}$. The procedure packs $L$ by
nonincreasing age and $P\setminus L$ by nonincreasing radius, using
groups of size at most $g$ in both cases. On $L$, the maximum radius of a group is at most
$\eta$ times its maximum age. On the complement, its maximum age is at
most $1/\eta$ times its maximum radius. The resulting cost is at most
$(1+\eta)D(L,t)+(1+1/\eta)S(P\setminus L)$, which is bounded
by~\eqref{eq:stream-biased-clearance}.
\end{proof}
Algorithm~\ref{alg:line-cleanup} implements the two clearance procedures.
\begin{algorithm}[!ht]
\caption{Clearing pending requests at an advice switch}\label{alg:line-cleanup}
\begin{algorithmic}[1]
\Function{SortPack}{$P,v,g$}
 \State Sort $P$ by nonincreasing $v_r$, breaking ties by request identity.
 \State \Return consecutive groups of $g$ requests, with a shorter last group if needed
\EndFunction
\Function{PackValue}{$P,v,g$}
 \State \Return $\displaystyle\sum_{G\in\Call{SortPack}{P,v,g}}\max_{r\in G}v_r$
\EndFunction
\Function{Clear}{$P,t,\textit{mode},\eta$}
 \State $w_r\gets t-a_r$ for $r\in P$
 \If{$\textit{mode}=\textit{ordinary}$}
  \State $\Pi\gets\Call{SortPack}{P,\max\{x,w\},g}$
 \Else
  \State $L\gets\{r\in P:x_r\le\eta w_r\}$
  \State $\Pi\gets\Call{SortPack}{L,w,g}\ \cup\
                  \Call{SortPack}{P\setminus L,x,g}$
 \EndIf
 \State Serve every nonempty $G\in\Pi$ at $t$; remove its members from physical pending.
\EndFunction
\end{algorithmic}
\end{algorithm}

\begin{proof}[Proof of Lemma~\ref{lem:stream-source-score}]
Arrivals and waiting only increase packing values. Subadditivity at a
service of $G$ gives
\[
 D(P_X,t)\le\max_{r\in G}w_r(t)+D(P_X\setminus G,t),
 \qquad S(P_X)\le\max_{r\in G}x_r+S(P_X\setminus G).
\]
The increase in $p_X$ therefore pays for the total decrease of $S+D$.
To bound $H_X$, we restrict all future services of $X$ to $P_X(t)$.
These restrictions form a feasible partition of the current virtual
pending set. Their total future cost is at
least the service cost plus the current-age cost of this partition,
which is at least $S(P_X(t))+D(P_X(t),t)$.

For~\eqref{eq:stream-spatial-certificate}, we partition $P$ using restrictions
of the past services of $B$, together with a partition of $P_B(t)$
minimizing service cost. Past service costs are at most $p_B(t)$, and restriction
cannot increase them. This proves the bound even when $P$ contains
requests that $B$ has already served.
\end{proof}
\subsection{Block budgets and geometric sums}
\begin{proof}[Proof of Lemma~\ref{lem:stream-block-budgets}]
Restrictions of source services cannot increase cost, and hence $F_J$ is at
most the source's paid-cost increment during the block. Just before a
source-changing event, the old active score is still within its cap;
the event causing a jump has not been committed. Monotonicity under
subsets gives $D(P_J,t_J)\le D(P_X,t_J)$ and $S(P_J)\le S(P_X)$.
For $A$, these facts bound $F_J+D(P_J,t_J)$ by $b_J$, proving
\eqref{eq:stream-A-block}. For $B$, apply
\eqref{eq:stream-ordinary-clearance} and bound paid cost plus twice
$S+D$ by twice the old cap $\delta b_J$. The same bounds hold without
exit cleanup for the final block.
\end{proof}

\begin{proof}[Proof of Lemma~\ref{lem:stream-level-accounting}]
Once $A$ is ineligible at a level, monotonicity of $H_A$ prevents it
from becoming eligible there again. Thus $B\to A$ requires a higher
level. The same observation proves the distinctness claims, including
when a block spans several levels without a cleanup.

Whenever a level is selected for $B$, preference for $A$ implies
$b\le H_A\le C_A$. This remains true for the last level of a $B$ block:
each same-source level increase again tests $A$ first. Each sum in
\eqref{eq:stream-advice-level-sums} is therefore a sum of distinct
dyadic levels at most $C_A$, and is at most $2C_A$. They are two separate
sums and will have separate coefficients in the cost bound.

Initially $b_0=2S(J)\le2C_B$. Every subsequent level is twice a level
where both eligibility tests failed, and hence $H_B\ge\delta b$ at that failed
level. Since $H_B\le C_B$ and $\delta\le1$, this proves
\eqref{eq:stream-largest-level}. The final assertion is a geometric sum.
\end{proof}
\subsection{Completion and implementation}
We finish the proof of Theorem~\ref{thm:feasible-stream-monitor} by
verifying completion and online computation. On a finite
input, $C_B$ is finite, and hence~\eqref{eq:stream-largest-level} permits only
finitely many levels and source changes. A final $B$ block completes
because $B$ does. A final $A$ block with a request pending forever would
have unbounded $D(P_A,t)$, exhausting its finite cap. Hence all requests
are eventually served without an end-of-input signal. The cost bounds
can first be applied to every finite prefix of physical actions, and hence
this completion argument does not assume either bound circularly.

Both packing values and both cleanups require only sorting. For fixed
$P_X$, ages have a fixed order and $D(P_X,t)$ is affine with slope
$\lceil|P_X|/g\rceil$. Timers therefore use only current requests.
Discrete jumps use the newly revealed event and the current source
states; intermediate ineligible dyadic levels can be skipped directly.
Under rational input encoding, including $\eps$, the number of levels and the arithmetic
overhead are polynomial whenever $B$ is polynomially simulable.
Malformed events are charged to the advice input length; as usual, event
streams are locally finite.

\subsection{Proof of the robustness lower bound}
\begin{proof}[Proof of Theorem~\ref{thm:stream-lower}]
We normalize the setup cost to one. The adversary releases the first
request at time zero and advises its immediate singleton service.
After the algorithm serves request $i$, the adversary waits
$\delta_i>0$ before releasing request $i+1$, again with immediate
singleton advice. The construction stops after request $g$.
Consistency forces the algorithm to serve each request even if no
others arrive.

Let $x_i$ denote the wait of request $i$, and write $X=\sum_{i<g}x_i$. On the input
stopped after $g-1$ requests, advice costs $g-1$, while the algorithm costs
at least $(g-1)+X$. Since future advice is not revealed, its execution
through these services is identical to that on the full input. Hence
$X\le\eps(g-1)$. On the full input, online cost is at least $g+X$.
All requests fit together, and hence serving them at the final arrival costs at
most $1+X+\Delta$, where $\Delta=\sum_{i<g}\delta_i$.
For $\Delta<g-1$, the function $(g+X)/(1+X+\Delta)$ decreases in $X$.
Letting $\Delta\downarrow0$ proves~\eqref{eq:stream-lower}; $g=1$ is
immediate. Letting $g\to\infty$ gives the stated lower bounds on
robustness uniform in capacity.
\end{proof}

\section{Proof of the arbitrary-size LAD extension}
\label{sec:arbitrary-sizes}
We prove the guarantee stated in Section~\ref{subsec:arbitrary-sizes}.

\begin{proof}[Proof of Corollary~\ref{cor:size-path-load}]
Under the preprocessing of Section~\ref{subsec:arbitrary-sizes}, every
closed small-item bin has load greater than $K/2$, since the rejected
item has size at most $K/2$. Large-item singleton batches also have
load greater than $K/2$.
Each preprocessed batch at radius $y$ costs at most $2y$. If $P$ is the
set of preprocessed requests, then
\[
 \ALG_{\rm pre}\le\frac4K\sum_{r\in P}s_ry_r.
\]
The residual cohort has total size at most $K$, even if its cardinality
is arbitrarily large.

For the weighted path load defined in Section~\ref{subsec:arbitrary-sizes},
the same integration argument gives $\sum_{r\text{ assigned to }u}s_r\le K$; whole cohorts,
and hence indivisible requests, remain intact. Service times and the
bound $3R_u$ are unchanged.

The weighted certificates are
\[
 \beta_r=\frac{s_rR_u}{K},\qquad
 \gamma_r=\frac{s_ry_v}{M_v}
 \quad\text{for an own-cohort request of a heavy root }v.
\]
All other requests receive zero $\gamma$ weight. If a heavy root has
$M_v=0$, we assign its whole weight $y_v$ to one designated request
in the cohort that created it. Such a root
is top-level: with a selected ancestor, a zero-mass cohort adds nothing
to a feasible path load and is accepted. Thus it cannot occur among the
nonmaximal heavy roots in the certificate proof.

For any feasible offline batch $B$, we write
$\alpha_B=K^{-1}\sum_{r\in B\setminus P}s_r\le1$.
Its preprocessed mass is at most $K(1-\alpha_B)$. The proofs of
Lemmas~\ref{lem:line-beta} and~\ref{lem:line-gamma} now give
\[
 \sum_{r\in B}\beta_r\le D_B+(2+\alpha_B)Y_B,
 \qquad
 \sum_{r\in B}\gamma_r\le D_B+(2+\alpha_B/\theta)Y_B.
\]
Indeed, the integral over each maximal represented home is unchanged;
the remaining terms are bounded by the total size $\alpha_BK$ in place
of a request count. Immediate child-root intervals are still disjoint,
and hence Lemma~\ref{lem:line-light} charges each request at most
$\beta_r/(1-\theta)$. Combining these bounds with preprocessing gives
exactly~\eqref{eq:line-direct-comparison}. The deterministic and
random-scale conclusions follow with the same constants. The operations
use cohort packing, tree ancestry, and arithmetic on revealed data;
no optimum bin packing is required.
\end{proof}

Algorithm~\ref{alg:size-path-load} replaces only the cohort
preprocessing and the load units in Algorithm~\ref{alg:path-load}.
It retains the final nonempty bin even when its load is exactly $K$.
\begin{algorithm}[H]
\caption{Weighted cohort preprocessing for the path-load algorithm}\label{alg:size-path-load}
\begin{algorithmic}[1]
\Require Home cohort $Q$ of width $y$, revealed at $b$; capacity $K$
\State Serve every $r\in Q$ with $s_r>K/2$ as a singleton at $b$.
\State $T\gets\varnothing$
\For{remaining requests $r$ in a fixed order}
 \If{$\sum_{q\in T}s_q+s_r>K$}
  \State Serve $T$ at $b$; $T\gets\varnothing$.
 \EndIf
 \State $T\gets T\cup\{r\}$
\EndFor
\If{$T\ne\varnothing$}
 \State Pass $T$ as one residual cohort to Algorithm~\ref{alg:path-load}.
 \State Use $M_v=\sum_{r\in T}s_r$ and replace $g$ by $K$ in every path-load term.
\EndIf
\State Leave assignments, ancestor selection, and root service times unchanged.
\State Retain a nonempty zero-load residual cohort; nonemptiness is not a load test.
\end{algorithmic}
\end{algorithm}

\section{Top-$k$ delay at one location}
\label{sec:top-q-tcp}

The fixed-location case also admits a constant guarantee when a batch pays
the sum of its largest waiting times. Retain the cardinality capacity
$g\in\mathbb Z_{\ge1}$ and let $k\in\mathbb Z_{\ge1}$ denote
the number of charged waits.
A batch of at most $g$ requests costs its fixed
setup $w>0$ plus $\operatorname{Top}_k$, the sum of its largest
$\min\{k,|Q|\}$ waiting times. The algorithm and comparator use the same
$g,k,w$. Requests arriving together retain their multiplicities.

\paragraph{The rule.}
Keep requests in arrival order. At each arrival epoch, serve full batches
of $g$ oldest pending requests until fewer than $g$ remain. Serve the
remaining buffer when its top-$k$ waiting sum first reaches $w$.
Full-batch actions precede a simultaneous timer. New arrivals have zero
age, so they do not jump the clock; full-batch actions clear the entire
old buffer before any newly created residual buffer can age.

\begin{theorem}[A uniform fixed-location bound]\label{thm:top-q-tcp}
The full-or-top-$k$ rule is strictly $3$-competitive for every $g,k\ge1$,
independently of both parameters. For $g\ge2$, no deterministic algorithm
has competitive ratio strictly below two.
\end{theorem}
\begin{proof}
Fix a feasible comparator $O$ and let $C_O$ be its total setup cost.
For each of its eventual batches $Q$, let $n_Q(t)$ count the requests
assigned to $Q$ that have arrived and are still unserved. Then
\begin{equation}\label{eq:top-q-tcp-integral}
 D_O^{(k)}=\int_0^\infty\sum_{Q\in O}\min\{k,n_Q(t)\}\,dt.
\end{equation}
This follows by integrating the waits of the oldest $k$ requests in
each batch.

The algorithm's positive-length buffer episodes have disjoint interiors.
Each ends by clearing the old buffer in a full batch or a timer batch.
All their delays are at most $w$, and timer batches have delay exactly $w$.
For comparison bookkeeping, reveal arrivals, perform all online full-batch
actions, then the comparator's actions, and finally any online timer.
Mark an episode when the comparator first serves a request currently in
its buffer. At a comparator action there is only one current buffer, and
no further buffer can start after the timer at that timestamp. Thus
marked episodes have distinct comparator batch witnesses, and their
total delay is at most $C_O$.

Throughout an episode ending unmarked, all buffer requests are still
pending offline. Partition them by their eventual comparator batches.
The inequality
\[
 \min\{k,\textstyle\sum_Q x_Q\}\le\sum_Q\min\{k,x_Q\}
\]
bounds the buffer's rate by the integrand
in~\eqref{eq:top-q-tcp-integral}. Disjoint episode interiors therefore
give $D_A^{(k)}\le C_O+D_O^{(k)}$. This also handles an episode cleared
by an arrival action whose requests are served offline at the same
endpoint: its open interval is still chargeable.

Let $F$ and $T$ count full and timer batches. Since full batches contain
$g$ distinct requests, $Fw\le Nw/g\le C_O$. Also $Tw\le D_A^{(k)}$.
Consequently,
\[
 \ALG=Fw+Tw+D_A^{(k)}
 \le C_O+2D_A^{(k)}
 \le3C_O+2D_O^{(k)}\le3\cost(O).
\]
A nonempty residual buffer has clock rate at least one and reaches its
timer in finite time. There are at most $N$ services, so completion
requires no end notification. If $g=1$, immediate service is optimal.

For $g\ge2$, the two offset adjacent-pair schedules in the proof of
Theorem~\ref{thm:fixed-k-two} remain feasible. Each offline pair is served
at its later arrival and has waits $(\Delta,0)$; each algorithm batch in
that adversarial construction is a singleton. Their costs are unchanged
for every $k$, so the deterministic lower bound of two follows.
\end{proof}

The old proof of a factor two cannot simply replace ``short span'' by
``small top-$k$ cost.'' For $g=3$, $k=2$, and $w=4$, arrivals
$(0,0,3,3,3,3)$ have three maximal consecutive prefixes with endpoint
top-$2$ cost at most four: $(0,0)$, $(3,3,3)$, and $(3)$.
Yet the two groups $(0,3,3)$ and $(0,3,3)$ both satisfy this condition.
Thus the corresponding greedy minimum-number lemma is false. The theorem
above concerns the fixed-location case; it does not change the general
LAD or advice guarantees elsewhere in the paper.

\section{Advice Experiment Protocol and Additional Results}\label{app:advice-experiments}
\subsection{Instances, capacities, and the exact benchmark}
We use 20 fixed seeds per family, indexed $0,\ldots,19$, and 12 requests per
instance. All arrivals and radii are quarter-integers. Before evaluation we
subtract the earliest arrival from every timestamp and sort requests by
arrival, retaining distinct identifiers at ties. The generator seed is
$2026111300+s+10000f$, where $s$ is the instance index and $f\in\{0,1,2\}$
is the family index. The same requests are evaluated at all three capacities.

\begin{center}
\begin{tabularx}{\textwidth}{@{}lXX@{}}
\toprule
Family & Unshifted arrival distribution & Radius distribution\\
\midrule
Bursty & Uniform center in $\{0,1,3,6\}$, plus $0$ or $1/4$ & Uniform in $\{1,2,4,8\}$\\
Spread-out & Uniform quarter-grid in $[0,32]$ & Uniform in $\{1/4,1,4,16\}$\\
Mixed-scale & Uniform quarter-grid in $[0,8]$ & Uniform in $\{2^j:-2\le j\le5\}$\\
\bottomrule
\end{tabularx}
\end{center}
Arrivals and radii are independently generated within these specified families.
These distributions are controlled synthetic tests, not calibrated demand models.

For a nonempty request subset $G$ with $|G|\le g$, its canonical batch cost is
\[
 c(G)=\max_{r\in G}x_r+\max_{r\in G}a_r-\min_{r\in G}a_r.
\]
It is attained by serving at the last arrival of $G$; later service only
increases cost. Thus exact offline optimization on these finite instances is
set partitioning. With $V(\varnothing)=0$, we compute
\[
 V(S)=\min_{\substack{\varnothing\ne G\subseteq S,\ |G|\le g\\
                     \min S\in G}}[c(G)+V(S\setminus G)].
\]
The restriction to the least identifier removes permutation duplicates.
Costs are computed as integers in quarter units, and the minimizing batches
are reconstructed. This exponential oracle is used only for the 12-request
benchmark; no polynomial-time exact algorithm for general capacitated LAD is
claimed. Independent restricted-growth partition enumeration agrees on 72
six-request/capacity tests. The single-location gap formula is checked on
two additional instances.

\subsection{Algorithms and information available}
The robust source is exactly path-load with fixed scale $\lambda=1$ and the
deterministic tie order in Algorithm~\ref{alg:path-load}. The monitor uses
$\eta=\eps/2$ and $\delta=\eps/[4(3+2/\eps)]$ without tuning.
We also record the independent-class full-or-age baseline from
Section~\ref{subsec:class-baseline}.

The experiment driver constructs complete advice and baseline event lists.
Decisions use only current virtual queues, paid costs, the event currently
revealed, and the endogenous score timer. Future advice is accessed only by
the driver and invariant checker. Arrivals precede virtual services at ties;
advice services precede baseline services; waiting timers are processed after
the timestamp's arrivals and services. Physical requests are served once,
in batches of size at most $g$. Cost includes all services after the final
arrival, without sending an end-of-input notification to the algorithm.

The exact-rational simulator is derived from the existing event audit.
For every run, it checks complete service, capacity, source-score monotonicity,
the residual-advice potential charge, and the stated consistency/robustness
inequalities. A separate cost evaluator recomputes costs from the recorded
physical schedule. An additional 111 prefix checks test that the path-load
schedule before a cutoff does not change when future input is deleted.
The six-request example in Figure~\ref{fig:path-load-example} has checked
algorithm cost $30$. These finite checks support implementation fidelity;
the universal guarantees are proved in the preceding appendices.

\subsection{Two advice corruption mechanisms}
\paragraph{Systematic lateness.}
Take an exact optimal schedule with $m$ batches and delay every service by
$\rho\OPT/m$. Its request partition is retained, so it remains complete and
capacity feasible, and its cost is exactly $(1+\rho)\OPT$. We use
\[
 \rho\in\{0,0.1,0.3,1,3,10,30,100,300,1000,10000\}.
\]
The large values intentionally expose the fallback regime of the conservative
monitor. They are not a claim about typical prediction errors. In total this
part has $180\times11\times3=5940$ runs.

\begin{figure}[ht]
\centering
\includegraphics[width=\textwidth]{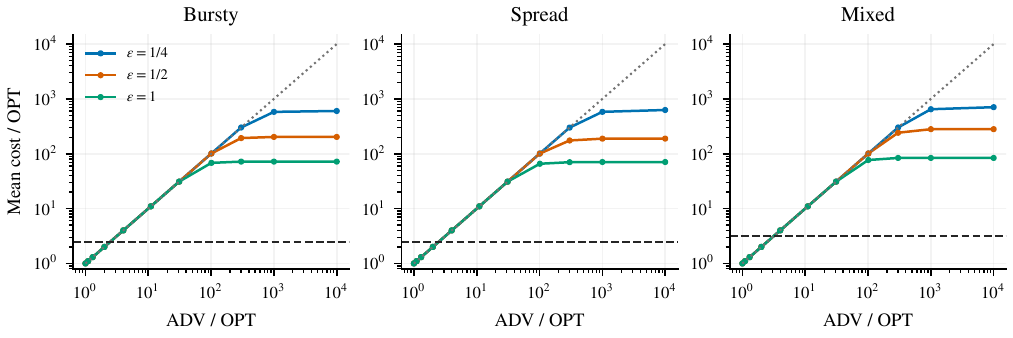}
\caption{All three input families separately, using the same instances and
parameters as the pooled main-text plot. Dotted gray is follow-advice;
dashed black is the path-load baseline. Each curve averages over 20 seeds
and three capacities. No family or seed is omitted.}
\label{fig:advice-families}
\end{figure}

\paragraph{Feasible repartitioning.}
For $\alpha\in\{0,0.2,0.4,0.6,0.8,1\}$, select $\lfloor\alpha m\rfloor$
optimal batches using a seeded random order. Shuffle their request identifiers
and form groups of at most $g$, serving each at its last arrival. Other batches
are retained. This changes the partition while maintaining feasibility and
completion. The partition shuffle seed is $910000+s+100f$, held fixed as
$\alpha$ changes. Advice quality need not be monotone in $\alpha$.
There are $180\times6\times3=3240$ runs. The realized advice-to-OPT ratio lies
between $1$ and $4.053$; every monitor has the same physical cost as following
advice and makes no nonempty switch in this regime. Figure~\ref{fig:advice-partition}
therefore shows their common cost, separately by capacity.

\begin{figure}[ht]
\centering
\includegraphics[width=.52\textwidth]{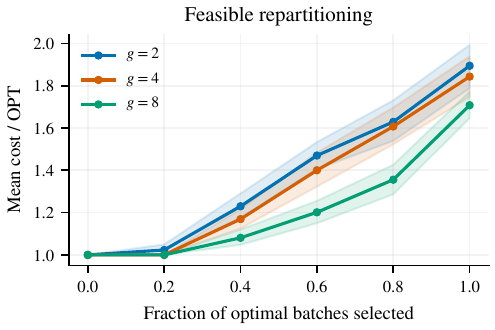}
\caption{Capacity-feasible partition corruption. All three monitor settings
coincide with follow-advice on these tests. Bands are pointwise 95\% stratified
bootstrap intervals over request-instance seeds.}
\label{fig:advice-partition}
\end{figure}

\subsection{Aggregation, uncertainty, and reproducibility}
The plotted estimator is the mean finite-input ratio $\ALG(I)/\OPT(I)$, not
a long-run stochastic ratio. In the pooled curves, we resample the 20 seeds
within each family, retaining all three capacity variants of a sampled
instance together. We use 2000 stratified paired cluster-bootstrap replicates
with seed $20261113$ and report pointwise percentile intervals. They are not
simultaneous confidence bands or evidence for untested input families.
Capacity-specific intervals use the corresponding fixed-capacity data.

Every generated input and all 9180 outcomes are included in the accompanying
\path{anc/experiments/results/advice_experiments.json}. The self-contained
\texttt{monitor\_core.py}, \texttt{run\_advice\_experiments.py}, and
\texttt{plot\_advice.py} reproduce the data and plots. Exact simulation uses
the Python standard library; plotting additionally uses NumPy and Matplotlib.
The unfavorable middle-corruption and poor-advice results are retained.
They identify conservative switching constants as a limitation of the present
implementation, alongside the benefit of optimal advice.

\paragraph{Additional information-flow audit.}
The script \path{anc/experiments/reviewer_audit.py} changes both future
arrivals and future advice while preserving each observed prefix. Across
540 paired runs, all physical actions and budget decisions through the
cutoff agree; these prefixes include 77 nonempty switches. It also rechecks
all 9180 exact cost records, the 180 optimal advice schedules, and the
independent small-instance oracle tests. This finite differential check
supports implementation fidelity and does not replace the online-information
argument in the proof. Lateness and repartitioning have different realized
error ranges, so their comparison does not isolate error structure at a
matched advice-to-OPT ratio. The original outcomes and algorithm constants
are unchanged.

\FloatBarrier

\bibliographystyle{plain}
\bibliography{references}
\end{document}